\newif\ifarxive
\arxivetrue 

\ifarxive
\documentclass[a4paper,USenglish,colorlinks=true,urlcolor=blue,citecolor=blue,linkcolor=red,UKenglish,runningheads]{lipics}

\newcommand{\subcaptionbox}[2]{\subfloat[#1]{#2}}
\else
\documentclass[a4paper,USenglish,colorlinks=true,urlcolor=blue,citecolor=blue,linkcolor=red,UKenglish,runningheads]{lipics-v2016}
\fi

\usepackage{microtype}
\usepackage{xspace}
\usepackage{color}
\usepackage{paralist}
\usepackage{amsmath}
\usepackage{amssymb}
\usepackage{todonotes}
\usepackage{wrapfig}
\usepackage{thm-restate}
\usepackage[%
    font={small,sf},
    labelfont=bf,
    format=hang,    
    format=plain,
    margin=0pt
]{caption}

\ifarxive
\newcommand{\lipicsarxiv}[2]{{#2}\xspace}
\else
\newcommand{\lipicsarxiv}[2]{{#1}\xspace}
\fi

\renewcommand{\paragraph}[1]{\smallskip\noindent\textbf{#1}\xspace}

\definecolor{proofcolor}{rgb}{0.3137254902,0.3176470588,0.3294117647}

\theoremstyle{plain}

\graphicspath{{figures/}}
\newtheorem{lemmax}{\textbf{Lemma}}

\newtheorem{property}{\textbf{Property}}

\newcommand{\skel}{\ensuremath{skel}}
\newcommand{\length}{\ensuremath{\ell}}

\newcommand{\journal}[1]{\xspace}
\def\etal{{\em et~al.}\xspace}

\author{Giordano {Da Lozzo}}
\author{William E. Devanny}
\author{David Eppstein}
\author{Timothy Johnson}

\affil{Computer Science Department, University of California, Irvine, USA\\
\texttt{\{gdalozzo,wdevanny,eppstein,tujohnso\}@uci.edu}}
\authorrunning{{Da Lozzo}~\etal} 

\Copyright{Giordano {Da Lozzo} and Will E. Devanny and David Eppstein and Timothy Johnson}


\ifarxive
\else
\EventEditors{Yoshio Okamoto and Takeshi Tokuyama}
\EventNoEds{2}
\EventLongTitle{28th International Symposium on Algorithms and Computation (ISAAC 2017)}
\EventShortTitle{ISAAC 2017}
\EventAcronym{ISAAC}
\EventYear{2017}
\EventDate{December 9--12, 2017}
\EventLocation{Phuket, Thailand}
\EventLogo{}
\SeriesVolume{92}
\ArticleNo{24} 
\fi

\definecolor{blue}{rgb}{0.274,0.392,0.666}
\definecolor{red}{rgb}{1,0.3,0.3}
\definecolor{green}{rgb}{0,0.588,0.509}

\newcommand{\red}[1]{{{\textcolor{red}{#1}\xspace}}}

\newtheorem{observation}{Observation}
\newcommand{\remove}[1]{}

\title{Square-Contact Representations of Partial $2$-Trees and Triconnected Simply-Nested Graphs\footnote{
Supported in part by the National Science Foundation under Grants CCF-1228639, CCF-1618301, and CCF-1616248.
This article also reports on work supported by the U.S.~Defense Advanced Research Projects Agency (DARPA) under agreement no.~AFRL FA8750-15-2-0092. The views expressed are those of the authors and do not reflect the official policy or position of the Department of Defense or the U.S.~Government.
}
} 
\titlerunning{Square-Contact Representations}

\newcommand{\scalednwarrow}{\ensuremath{\text{\scalebox{0.5}{$\nwarrow$}}}}
\newcommand{\scalednearrow}{\ensuremath{\text{\scalebox{0.5}{$\nearrow$}}}}
\newcommand{\scaledswarrow}{\ensuremath{\text{\scalebox{0.5}{$\swarrow$}}}}
\newcommand{\scaledsearrow}{\ensuremath{\text{\scalebox{0.5}{$\searrow$}}}}

\newcommand{\scaledsymbol}[1]{\ensuremath{\text{\scalebox{0.7}{$#1$}}}}
\newcommand{\Scalednwarrow}{\scaledsymbol{\nwarrow}}
\newcommand{\Scalednearrow}{\scaledsymbol{\nearrow}}
\newcommand{\Scaledswarrow}{\scaledsymbol{\swarrow}}
\newcommand{\Scaledsearrow}{\scaledsymbol{\searrow}}

\newcommand{\UR}{\ensuremath{\Scalednearrow}}
\newcommand{\UL}{\ensuremath{\Scalednwarrow}}
\newcommand{\DR}{\ensuremath{\Scaledsearrow}}
\newcommand{\DL}{\ensuremath{\Scaledswarrow}}

\def\NUR(#1){\ensuremath{\overset{\scalednearrow}{#1}}}
\def\NUL(#1){\ensuremath{\overset{\scalednwarrow}{#1}}}
\def\NDR(#1){\ensuremath{\overset{\scaledsearrow}{#1}}}
\def\NDL(#1){\ensuremath{\overset{\scaledswarrow}{#1}}}
\def\NPAR(#1,#2){\ensuremath{\overset{\scaledsymbol{#2}}{#1}}}

\newcommand{\hatsymbol}{\circ}

\AtBeginDocument{%

}

\begin{document}
\maketitle

\begin{abstract}
A {\em square-contact representation} of a planar graph $G=(V,E)$ maps vertices in $V$ to interior-disjoint axis-aligned squares in the plane and edges in $E$ to adjacencies between the sides of the corresponding squares.~In this paper, we study {\em proper} square-contact representations of planar graphs, in which any two squares are either disjoint or share infinitely many points.  

We characterize the partial $2$-trees and the triconnected cycle-trees allowing for such representations. For partial $2$-trees our characterization uses a simple forbidden subgraph whose structure forces a separating triangle in any embedding. For the triconnected cycle-trees, a subclass of the triconnected simply-nested graphs, we use a new structural decomposition for the graphs in this family, which may be of independent interest.
Finally, we study square-contact representations of general triconnected simply-nested graphs with respect to their \mbox{outerplanarity index.}
\end{abstract}

\section{Introduction}\label{se:intro}

Contact representations of graphs, in which the vertices of a graph are represented by non-overlapping or non-crossing geometric objects of a specific type, and edges are represented by tangencies or other contacts between these objects, form an important line of research in graph drawing and geometric graph theory. For instance, the Koebe--Andreev--Thurston circle packing theorem states that every planar graph is a contact graph of circles~\cite{Ste-CUP-05}. Other types of contact representations that have been studied include contacts of unit circles~\cite{BowDurLof-GD-15,Har-EdM-74}, line segments~\cite{Hli-DM-01}, circular arcs~\cite{AlaEppKau-WADS-15}, triangles~\cite{GonLevPin-DCG-12}, L-shaped polylines~\cite{ChaKobUec-WG-13}, and cubes~\cite{FelFra-SoCG-11}.

Schramm's monster packing theorem~\cite{Sch-PhD-90} implies that every planar graph can be represented by the tangencies of translated and scaled copies of any smooth convex body in the plane. However, it is more difficult to use this theorem for non-smooth shapes, such as polygons: when $k$ bodies can meet at a point, the  monster theorem may pack them in a degenerate way in which separating $k$-cycles, and their interiors, shrink to a single point.

In this paper we study one of the simplest cases of contact representations that cannot be adequately handled using the monster theorem: contact systems of axis-parallel squares. We distinguish between \emph{proper} and \emph{improper} contacts: a proper contact representation disallows squares that meet only at their corners, while an \emph{improper} or \emph{weak} contact representation allows corner-corner contacts of squares. These weak contacts may represent edges of the graph, but they are also allowed between squares that should be non-adjacent. The weak contact representations by squares were shown by Schramm~\cite{Sch-IJM-93} to include all of the proper induced subgraphs of maximal planar graphs that have no separating 3-cycles or 4-cycles. However, a characterization of the graphs \mbox{having proper contact representations by squares remains elusive.}

There is a simple necessary condition for the existence of a proper contact representation by squares. No three properly-touching squares can surround a nonzero-area region of the plane. Therefore, if every embedding of a planar graph $G$ with four or more vertices has a separating triangle or a triangle as the outer face, then $G$ cannot have a proper contact representation. Our main results show that this necessary condition is also sufficient for two notable families of planar graphs: partial 2-trees (including series-parallel graphs) and triconnected cycle-trees (including the Halin graphs). 
However, we show that this necessary condition is not sufficient for the existence of weak and proper square-contact representations of $3$-outerplanar and $2$-outerplanar triconnected simply-nested graphs.

Due to space limits, full versions of omitted or sketched proofs \mbox{are provided \lipicsarxiv{in~\cite{paperarxiv}.}{in~\autoref{apx:proofs}.}}

\section{Preliminaries} 

For standard graph theory concepts and definitions related to planar graphs, their embeddings, and connectivity we refer the reader, e.g., to~\cite{gdbook} and \lipicsarxiv{to~\cite{paperarxiv}.}{to~\autoref{apx:definitions}.}

The graphs considered in this paper are planar, finite, simple, and connected.
We denote the vertex set $V$ and the edge set $E$ of a graph $G=(V,E)$ by $V(G)$ and $E(G)$, respectively. Let $H$ and $G$ be two graphs. We say that $G$ is $H$-free if $G$ does not contain a subgraph isomorphic to $H$. 
The complete $k$-partite graph $K_{|V_1|,\dots,|V_k|}$ is the graph $(V = \bigcup^k_{i=1} V_i, E = \bigcup_{i<j} V_i \times V_j)$.

\paragraph{\textbf{Series-parallel graphs and partial 2-trees.}}
A {\em two-terminal series-parallel} graph $G$ with source $s$ and target $t$ can be recursively defined as follows: 
\begin{inparaenum}[(i)]
\item Edge $st$ is a two-terminal series-parallel graph. 
Let $G_1,\dots,G_k$ be two-terminal series-parallel graphs and let $s_i$ and $t_i$ be the source and the target of $G_i$, respectively, with $1 \leq i \leq k$. 
\item The {\em series composition} of $G_1,\dots,G_k$ obtained by identifying $s_i$ with $t_{i+1}$, for $i=1,\dots,k-1$, is a two-terminal series-parallel graph with source $s_k$ and target $t_1$; and 
\item the {\em parallel composition} of $G_1,\dots,G_k$ obtained by identifying $s_i$ with $s_1$ and $t_i$ with $t_1$, for $i=2,\dots,k$, is a two-terminal series-parallel graph with source $s_1$ and target $t_1$.
\end{inparaenum}

A {\em series-parallel graph} is either a single edge or a two-terminal series-parallel graph with the addition of an edge, called {\em reference edge} joining $s$ and $t$. Clearly, series-parallel graphs are $2$-connected.
%
A series-parallel graph $G$ with reference edge $e$ is naturally associated with a rooted 
tree $T$, called the {\em SPQ}-tree of $G$. Each internal node of $T$, with the exception of the one associated with $e$, corresponds to a two-terminal series-parallel graph. Nodes of $T$ are of three types: {\em S-}, {\em P-}, and {\em Q-nodes}.
Further, tree $T$ is rooted to the Q-node corresponding to $e$.
%

Let $\mu$ be a node of $T$ with terminals $s$ and $t$ and children $\mu_1,\dots,\mu_k$, if any. 
Node $\mu$ has an associated 
multigraph, called the \emph{skeleton} of $\mu$ and denoted by $\skel_\mu$, containing a {\em virtual edge} $e_i=s_it_i$, for each child $\mu_i$ of $\mu$.
 Skeleton $\skel_\mu$ shows how the children of $\mu$, represented by ``virtual edges'', are arranged into $\mu$. 
The skeleton $\skel_\mu$ of $\mu$ is:
\begin{inparaenum}[(i)]
\item edge $st$, if $\mu$ is a leaf Q-node,
\item the multi-edge obtained by identifying the source $s_i$ and the target $t_i$ of each virtual edge $e_i$, for $i=1,\dots,k$, with a new source $s$ and and new target $t$, respectively, or
\item the path $e_1, \dots, e_k$, where virtual edge $e_i$ and $e_{i+1}$ share vertex $s_i=t_{i+1}$, with $1 \leq i < k$.
\end{inparaenum} 
If $\mu$ is an S-node, then we denote by $\length(\mu)$ the length of $skel_\mu$, i.e., $\ell(\mu)=k$.

Also, for each virtual edge $e_i$ of $\skel_\mu$, recursively replace $e_i$ with the skeleton $\skel_{\mu_i}$ of its corresponding child $\mu_i$. The two-terminal series-parallel subgraph of $G$ that is obtained in this way is the \emph{pertinent graph} of $\mu$ and is denoted by $G_\mu$.
We have that $G_\mu$ is:
\begin{inparaenum}[(i)]
\item edge $st$, if $\mu$ is a Q-node,
\item the series composition of the two-terminal series-parallel graphs $G_{\mu_1},\dots,G_{\mu_k}$, if $\mu$ is an S-node, and
\item the parallel composition of the two-terminal series-parallel graphs $G_{\mu_1},\dots,G_{\mu_k}$, if $\mu$ is a P-node.
\end{inparaenum} 
We denote by 
$G^-_\mu$ the subgraph of $G_\mu$ obtained by removing from it terminals $s$ and $t$ together with their incident edges.

A {\em $2$-tree} is a graph that can be obtained from an edge by repeatedly adding a new vertex connected to two adjacent vertices. 
Every $2$-tree is  planar and $2$-connected. 
A {\em partial $2$-tree} is a subgraph of a $2$-tree. Equivalently, {\em partial $2$-tree} can be defined as the $K_4$-minor-free graphs. 
In particular, the series-parallel graphs are
exactly the $2$-connected partial $2$-trees.

\paragraph{\textbf{Simply-nested graphs.}}
Let $G$ be an embedded planar graph and let $G_1,\dots,G_k$ be the sequence of embedded planar graphs such that $G_1=G$, $G_{i+1}$ is obtained from $G_i$ be removing all the vertices incident to the outer face of $G_i$ together with their incident edges, and $G_k$ is outerplanar. 
We say that the embedding of $G$ is {\em $k$-outerplanar}. A graph is {\em $k$-outerplanar} if it admits a
$k$-outerplanar embedding. 
The set $V_i$ of vertices incident to the outer face of $G_i$ is the {\em $i$-th level} of $G$. 
A $k$-outerplanar graph is {\em simply-nested}~\cite{c-fhccpg-90} if, for $i=1,\dots,k-1$, graphs $G[V_i]$ are chordless cycles and $G[V_k]$ is either a cycle or a tree. 

We define {\em cycle-trees} and {\em cycle-cycles} the $2$-outerplanar simply-nested graphs whose internal \mbox{level is a tree and a cycle, respectively.}
The $2$-outerplanar $3$-connected simply-nested graphs have a nice geometric interpretation.
Similarly to the Halin graphs, which are the graphs of polyhedra containing a face that share an edge with all other faces, 
$3$-connected cycle-trees are the graphs of polyhedra containing a face touched by all other faces. 
Analogously, the $3$-connected cycle-cycle graphs with no chords on the inner cycle are the graphs of polyhedra in which there exist two disjoint faces that are both touched by all other faces.

\paragraph{\textbf{Square-contact representations.}} Let $G=(V,E)$ be a planar graph. A {\em square-contact representation} $\Gamma$ of $G$ maps each vertex $v \in V$ to an axis-aligned square $S_{\Gamma}(v)$ in the plane, such that, for any two vertices $u,v \in V$, squares $S_{\Gamma}(u)$ and $S_{\Gamma}(v)$ are interior-disjoint, and the sides of $S_{\Gamma}(u)$ and $S_{\Gamma}(v)$ touch if and only if $uv \in E$. A square-contact representation of $G$ is {\em proper} if any two touching squares share infinitely many points, i.e., they cannot share only a corner point, and {\em non-proper}, otherwise. When the square-contact representation is clear from the context, we may choose to drop the $\Gamma$ subscript and just use $S(v)$ to refer to the square for vertex $v$.
In the remainder of the paper, we only consider proper square-contact representations and refer to such representations simply as square-contact representations.

\paragraph{\textbf{Geometric transformations.}} Let $G$ be planar graph and let  $\Gamma$ be a square-contact representation of $G$. Also, let $p$ be any point in $\Gamma$. 
We define the $\UR$-, $\UL$-, $\DL$-, and $\DR$-{\em quadrant} of $p$ in $\Gamma$ as the first, second, third, and fourth
quadrant around $p$, respectively. 
Suppose that the half-lines delimiting the $\DL$-quadrant of $p$ in $\Gamma$ do not intersect the interior of any square in $\Gamma$. Also, let $\Gamma'$ be the part of $\Gamma$ lying in the $\DL$-quadrant of $p$.
Then, a {\em \NDL(p)-scaling} of $\Gamma$ {\em by a factor $\alpha>0$} is a square-contact representation $\Gamma^*$ defined as follows; see, e.g., \autoref{fig:finaledge}.
Initialize $\Gamma^*=\Gamma$ and remove from $\Gamma^*$ the drawing of the squares contained in the interior of $\Gamma'$. Then, insert into $\Gamma^*$ a copy $\Gamma''$ of $\Gamma'$ scaled by $\alpha$ such that the upper-right corner of $\Gamma''$ coincides with $p$. Clearly, depending on the scale factor $\alpha$, drawing $\Gamma^*$ may or may not be a square-contact representation of $G$ (as adjacencies may be lost or gained). In the following, we refer to the case in which $\alpha>1$ simply as a {\em \NDL(p)-scaling} of $\Gamma$ and to the case in which $0 < \alpha < 1$ as a {\em negative \NDL(p)-scaling} of $\Gamma$.
The definition of {\em \NPAR(p,\hatsymbol)-scaling} and {\em negative \NPAR(p,\hatsymbol)-scaling}, 
with $\hatsymbol \in \{\UL,\DR,\UR\}$, is analogous. 
Finally, let $v$ be a vertex of $G$ and let $x$, $y$, $z$, and $w$ be the upper-left, lower-left, lower-right, and upper-left corner points of $S(v)$ in $\Gamma$. A {\em \NUL(v)-scaling}, {\em \NDL(v)-scaling}, {\em \NDR(v)-scaling}, {\em \NUR(v)-scaling} of $\Gamma$ is a {\em \NUL(x)-scaling}, {\em \NDL(y)-scaling}, \mbox{{\em \NDR(z)-scaling}, {\em \NUR(w)-scaling} of $\Gamma$, respectively.}

\section{Partial 2-Trees}\label{se:sp}
In this section, we study square-contact representations of partial $2$-trees and give the following simple characterization for graphs in this family admitting such representations.

\begin{theorem}\label{th:main}
Let $G$ be a partial $2$-tree. Then, the following statements are equivalent:
\begin{enumerate}[(i)]
\item $G$ is $K_{1,1,3}$-free,
\item $G$ admits an embedding without separating triangles, and
\item $G$ admits a square-contact representation.
\end{enumerate}
\end{theorem}

In order to prove \autoref{th:main}, we first show that, without loss of generality, we can restrict our attention to the biconnected partial 2-trees, i.e., the series-parallel graphs.

\begin{restatable}{lemmax}{Lemmatwotreetosp}\label{le:split}
Let $G$ be a $K_{1,1,3}$-free partial $2$-tree. Then, there exists a $K_{1,1,3}$-free series-parallel graph $G^*$ such that $G \subset G^*$ 
and $G$ admits a square-contact representation if~$G^*$~does.
\end{restatable}

\newcommand{\proofoflemmatwotreetosp}{
Let $\beta(H)$ denote the number of blocks of a graph $H$.
We show that $G$ can be augmented to a $K_{1,1,3}$-free partial $2$-tree $G'$ such that $\beta(G')=\beta(G)-1$, by adding to $G$ a new vertex connected to two vertices in $V(G)$ incident to the same cut-vertex of $G$, belonging to different blocks, and sharing a common face. Hence, repeating such an augmentation eventually yields a $2$-connected partial $2$-tree $G^*$ that is $K_{1,1,3}$-free. Also, by construction, two vertices in $V(G)$ are adjacent in $G^*$ if and only if they are adjacent in $G$. Therefore, a square-contact representation of $G$ can be  derived from a square-contact representation $\Gamma^*$ of $G^*$, by removing from $\Gamma^*$ all the squares corresponding to vertices in $V(G^*) \setminus V(G)$.

Suppose that $\beta(G)>1$, as otherwise $G$ is $2$-connected and we can set $G^*=G$. Consider a planar embedding $\cal E$ of $G$ and a cut-vertex $c$ of $G$. Let $e_1=(c,u)$ and $e_2=(c,v)$ be two edges incident to $c$ such that (i) $e_1$ and $e_2$ belong to distinct blocks $B_1$ and $B_2$ of $G$, respectively, and (ii) $e_1$ precedes $e_2$ in the clockwise order of the edges incident to $c$ in $\cal E$. Let $f$ be the face of $\cal E$ that lies to the right of $e_1$ (to the left of $e_2$) when traversing $e_1$ ($e_2$) from $c$ to $u$ (from $c$ to $v$). Augment $G$ to graph $G'$ by adding a new vertex $w$ and edges $e_1'=(w,u)$ and $e_2'=(w,v)$ inside $f$.
Clearly, blocks $B_1$ and $B_2$ of $G$ are now ``merged'' in a single block $B$ in $G'$. Also, $G'$ is a partial $2$-tree. In fact, since $G$ is a partial $2$-teee and since any path connecting a vertex in $V(B_1)$ and a vertex in $V(B_2)$ must pass either through $c$ or trough $w$, graph $G'$ does not contain $K_4$ as a minor. Finally, $G'$ is  $K_{1,1,3}$-free. This is due to the fact that
$G$ is $K_{1,1,3}$-free and that, since $u$ and $v$ belong to different blocks of $G$, these vertices do not belong to an induced $K_{1,1,2}$ subgraph of $G$.
\xspace{}
}


\begin{proof}[Sketch]
Let $\beta(H)$ denote the number of blocks, i.e., the maximal biconnected components, of a graph $H$.
Adding to $G$ a new vertex connected to two vertices in $V(G)$ incident to the same cut-vertex of $G$, belonging to different blocks, and sharing a common face
yields a graph $G'$ such that $\beta(G')=\beta(G)-1$. It is easy to see that $G'$ is $K_{1,1,3}$-free and that $G'$ does not contain $K_4$ as a minor.
Hence, repeating such an augmentation eventually yields a $2$-connected partial $2$-tree $G^*$ that is $K_{1,1,3}$-free. Also, by construction, two vertices in $V(G)$ are adjacent in $G^*$ if and only if they are adjacent in $G$. Therefore, a square-contact representation of $G$ can be  derived from a square-contact representation $\Gamma^*$ of $G^*$, by removing from $\Gamma^*$ all the squares corresponding to vertices in $V(G^*) \setminus V(G)$.
\end{proof}

As already observed in \autoref{se:intro},  an embedding without separating triangles is  necessary  for the existence of a square-contact representation, and $K_{1,1,3}$ has no embedding without separating triangles. Thus,
$(iii)\Rightarrow (ii) \Rightarrow(i)$ are immediate. To complete the proof of \autoref{th:main}, we show how to construct a square-contact representation of any $K_{1,1,3}$-free series-parallel graph, proving that  $(i) \Rightarrow (iii)$. We formalize this result in the next theorem.

\begin{theorem}\label{th:construction}
Every $K_{1,1,3}$-free series-parallel graph admits a square-contact representation.
\end{theorem}

\begin{figure}[t!]
   \centering
   \subcaptionbox{\label{fig:structure-a}}
  {\includegraphics[height=.2\columnwidth,page=1]{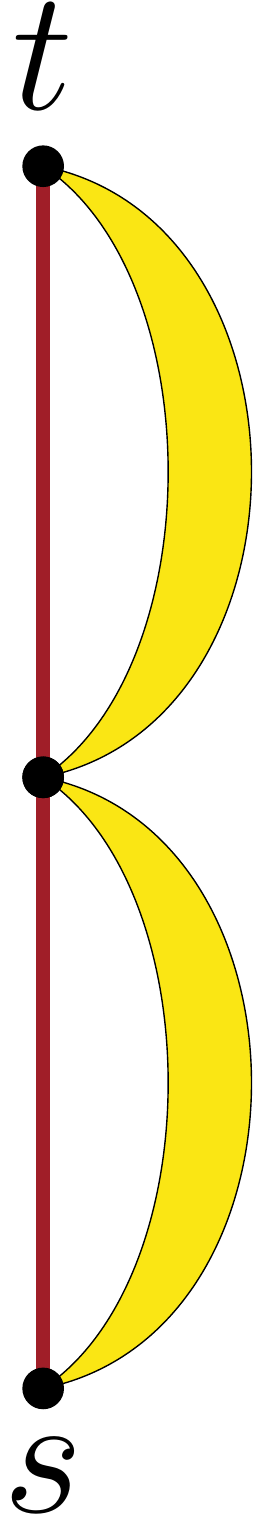}}
  \hfil
  \subcaptionbox{\label{fig:structure-b}}
  {\includegraphics[height=.2\columnwidth,page=3]{structure}}
  \hfil
  \subcaptionbox{\label{fig:structure-c}}
  {\includegraphics[height=.2\columnwidth,page=4]{structure}}
  \hfil
    \subcaptionbox{\label{fig:structure-d}}
  {\includegraphics[height=.2\columnwidth,page=6]{structure}}
    \hfil
  \subcaptionbox{\label{fig:structure-e}}
  {\includegraphics[height=.2\columnwidth,page=5]{structure}}
  \hfil
  \subcaptionbox{\label{fig:structure-f}}
  {\includegraphics[height=.2\columnwidth,page=7]{structure}}
  \caption{(a) A critical S-node, (b) an almost-bad P-node, (c) a bad P-node, (d) a forbidden P-node, (e) an S-node of \texttt{Type B}, and (f) an S-node of \texttt{Type C}. Yellow, green, and blue regions, represent parallel compositions of any number of S-nodes, at most one critical S-node and any number of non-critical S-nodes, and any number of non-critical S-nodes, respectively.}
  \label{fig:structure}
\end{figure}

Let $G$ be a series-parallel graph and let $T$ be the SPQ-tree of $G$ with respect to any reference edge.
We start with some definitions; refer to \autoref{fig:structure}.
Let $\mu$ be an S-node in $T$.
We say that $\mu$ is
{\em critical},  if $\skel_\mu=$$s$--$x$--$t$ and the two children of $\mu$ both contain an edge between their terminals, i.e., $sx, xt \in E(G_\mu)$, and 
{\em non-critical}, otherwise.
Let $\mu$ be a P-node in $T$ containing an edge between its terminals. We say that $\mu$ is
{\em almost bad}, if it has exactly one critical child, 
{\em bad}, if it has exactly two critical children, and
{\em forbidden}, if it has more than two critical children.
Finally, let $\mu$ be a P-node in $T$. We say that $\mu$ is 
{\em good}, if it is neither bad, nor almost bad, nor forbidden.

We now assign one of three possible types to each S-node $\mu$ in $T$ as follows (for each child $\mu_i$ of $\mu$, we denote the two terminals of $G_{\mu_i}$ as $s_i$ and $t_i$).

\begin{description}
\item[Type A] Node $\mu$ is of \texttt{Type A}, if either $\length(\mu)>2$ or $\length(\mu)=2$ and at least one child of $\mu$ does not contain an edge between its terminals, i.e., $|\{s_1t_1,s_2t_2\} \cap E(G_\mu)| < 2$.
\item[Type B] Node $\mu$ is of \texttt{Type B}, if $\length(\mu)=2$, 
all its children contain an edge between their terminals, and at least one of them is a bad P-node.
\item[Type C] Node $\mu$ is of \texttt{Type C}, if $\length(\mu)=2$, and all its children contain an edge between their terminals, and none of them is a bad P-node. 
\end{description}

Observe that S-nodes of \texttt{Type B} and of \texttt{Type C} are also critical.

Let $G$ be a $K_{1,1,3}$-free series-parallel graph and let $T$ be the SPQ-tree of $G$ with respect to any reference edge.  We have the following simple observations regarding the P-nodes in $T$.

\begin{observation}\label{obs:no-forbidden}
SPQ-tree $T$ contains no forbidden P-node; refer to \red{\autoref{fig:structure}(d)}.
\end{observation}

\begin{observation}\label{obs:edge-implications}
Let $\mu$ be a P-node in $T$ with terminals $s$ and $t$ such that 
$st \in E(G_\mu)$. Then, none of the children of $\mu$ is of \texttt{Type B} and at most two children of $\mu$ are of \texttt{Type C}. 
\end{observation}

We now consider special square-contact representations for the pertinent graphs of the S-nodes in $T$.
Let $\Gamma_\mu$ be a square-contact representation of $G_\mu$. We say that $\Gamma_\mu$ is either an {\em L-shape}, {\em pipe}, or {\em rectangular drawing} of $G_\mu$, if it satisfies the \mbox{following conditions; refer to \autoref{fig:drawings}.}

\begin{description}
\item[Rectangular drawing] 
$S(t)$ lies to the left and above $S(s)$ and
the drawing $\Gamma^-_\mu$ of $G^-_\mu$ in $\Gamma_\mu$ lies to the right of $S(t)$ and above $S(s)$; also, all the squares of $\Gamma^-_\mu$ whose left side (bottom side) is collinear with the right side of $S(t)$ (with the top side of $S(s)$) are adjacent to $S(t)$~(to~$S(s)$).
\item[L-shape drawing] $\Gamma_\mu$ is a rectangular drawing in which 
there exists a rectangular region (red region $R_\emptyset$ in \autoref{fig:drawings}) inside the bounding box of $\Gamma^-_\mu$ whose interior does not intersect any square in $\Gamma^-_\mu$ and whose lower-left corner lies at the intersection point between the vertical line passing through the right side of $S(t)$ and the horizontal line passing through the top side of~$S(s)$.
\item[Pipe drawing] 
$S(t)$ lies to the left of $S(s)$ and
the drawing $\Gamma^-_\mu$ of $G^-_\mu$ in $\Gamma_\mu$ lies to the right of $S(t)$ and to the left of $S(s)$; also, all the squares of $\Gamma^-_\mu$ whose left side (right side) is collinear with the right side of $S(t)$ (with the left side of $S(s)$) are \mbox{adjacent to $S(t)$ (to $S(s)$).}
\end{description}

\begin{figure}[t!]
  \centering
  \includegraphics[height=.2\columnwidth,page=4]{l-pipe-rectangular}
  \hfil
  \includegraphics[height=.2\columnwidth,page=1]{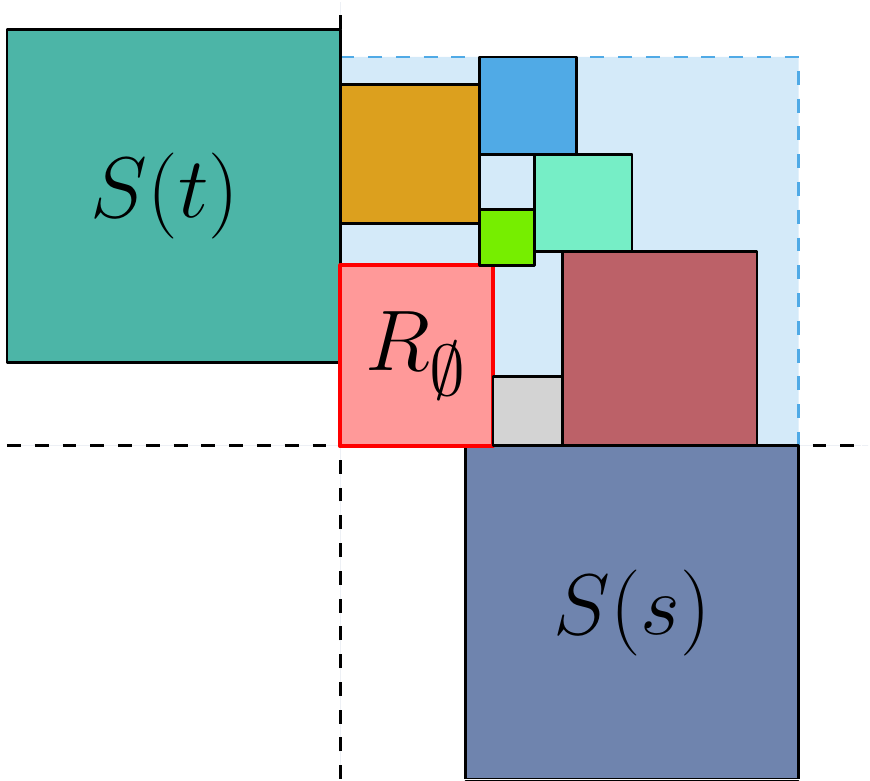}
    \hfil
  \includegraphics[height=.2\columnwidth,page=2]{l-pipe-rectangular}
  \hfil
  \includegraphics[height=.2\columnwidth,page=3]{l-pipe-rectangular}
  \caption{
  From left to right:
  pertinent $G_\mu$ of an S-node $\mu$ with terminals $s$ and $t$,
  L-shape and pipe drawings of $G_\mu$, respectively, 
  and a rectangular drawing of an S-node $\nu$ with pertinent $G_\nu = G_\mu \cup sx$. The L-shape region and horizontal pipe enclosing $G^-_\mu$ and the rectangle enclosing $G^-_\nu$ are shaded~blue.}
  \label{fig:drawings}
\end{figure}

In the following, we generally refer to a drawing of an S-node $\mu$ in $T$ (of $G_\mu$) which is either an L-shape drawing, a pipe drawing, or a rectangular drawing as a \mbox{{\em valid drawing} of $\mu$ (of $G_\mu$).}

Let $\Gamma^-_\mu$ be the square-contact representation of $G^-_\mu$ contained in $\Gamma_\mu$.
Observe that $\Gamma^-_\mu$ lies in the interior of an orthogonal hexagon with an internal angle equal to $270^\circ$, i.e., an {\em L-shape polygon} (or, simply, {\em L-shape}), if $\Gamma^-_\mu$ is an L-shape drawing.
Also, $\Gamma^-_\mu$ lies \mbox{in the interior of a} rectangle whose opposite vertical sides are adjacent to the right side of $S(t)$ and to the left side of~$S(s)$, i.e., a {\em horizontal pipe}, if $\Gamma^-_\mu$ is a pipe drawing. Finally, $\Gamma^-_\mu$ lies in the \mbox{interior of a} rectangle whose left and bottom side are adjacent to the right side of $S(t)$ and \mbox{to the top side} of $S(s)$, respectively, if $\Gamma^-_\mu$ is a rectangular drawing.

\smallskip
\begin{wrapfigure}[11]{R}{.38\textwidth}
    \centering
    \ifarxive
\vspace{-0.2mm}
\else
\vspace{-6.2mm}
\fi
    \includegraphics[page=13, width=.39\textwidth]{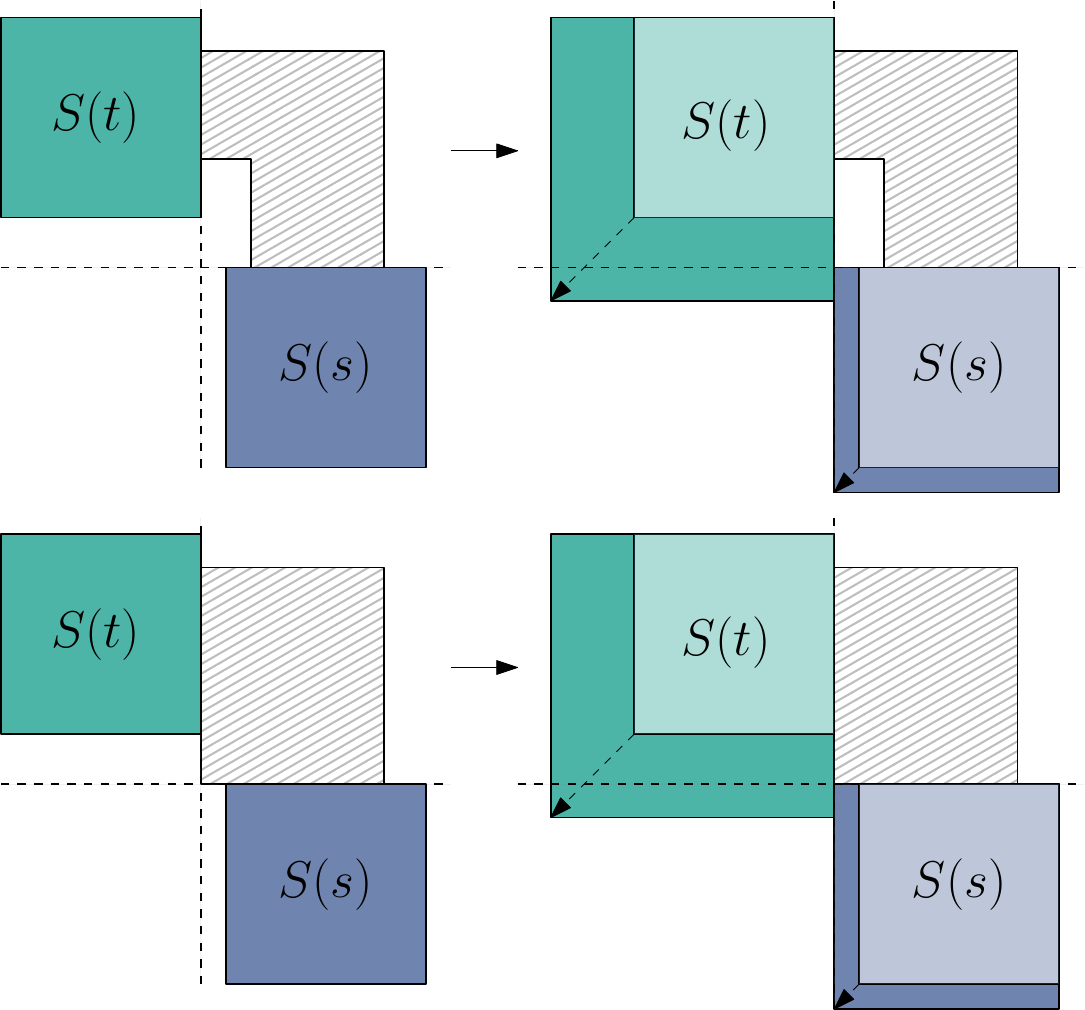}
    \caption{Transforming $\Gamma_\tau$ into $\Gamma_\rho$.}\label{fig:finaledge}
\end{wrapfigure}
\paragraph{Proof of \autoref{th:construction}.} 
In order to prove \autoref{th:construction}, we proceed as follows. 
Let $G$ be a $K_{1,1,3}$-free series-parallel graph and let $T$ be the SPQ-tree of $G$ rooted at a Q-node $\rho$ with terminals $s$ and $t$, whose unique child $\tau$ is an S-node. Observe that such a Q-node always exists, since $G$ is simple, and that node $\tau$ is either of \texttt{Type A} or of \texttt{Type C}, since $G$ is $K_{1,1,3}$-tree. 
We perform a bottom-up traversal in $T$ to construct one or two valid drawings of $G_\mu$, for each S-node $\mu \in T$. \mbox{Namely, we compute:}
\begin{itemize}
	\item an L-shape drawing, if $\mu$ is of \texttt{Type A} (\autoref{le:A}),
	\item a pipe drawing, if $\mu$ is of \texttt{Type B} (\autoref{le:B}), and
	\item both a pipe drawing and a rectangular drawing, if $\mu$ is of \texttt{Type C} (\autoref{le:C}).
\end{itemize}

Thus, when node $\tau$ is considered, we can compute either an L-shape drawing of $G_\tau$, if $\tau$ is of \texttt{Type A}, or a rectangular drawing of $G_\tau$, if $\tau$ is of \texttt{Type C}. 
Further, both such valid drawings $\Gamma_\tau$ of $G_\tau$ can be easily turned into a square-contact representation $\Gamma_{\rho}$ of $G=G_\tau \cup st$, by performing a {\NDL(t)-scaling} and an {\NDL(s)-scaling} of $\Gamma_\tau$ 
in such a way that the right side of $S(t)$ and the left side of $S(s)$ touch; refer to \autoref{fig:finaledge}. This is possible since both in an L-shape drawing and in a rectangular drawing of $G_\tau$ all the squares of $G^-_\tau$ whose left side (bottom side) is collinear with the right side of $S(t)$ (with the top side \mbox{of $S(s)$) are adjacent to $S(t)$ (to $S(s)$).}

Let $\mu$ be an S-node and let $\mu_1,\dots,\mu_k$ be the children of $\mu$ in $T$. 
If each child $\mu_i$ of $\mu$ is a Q-node, then node $\mu$ is of \texttt{Type A}, if $\length(\mu)>2$, and it is of \texttt{Type C}, otherwise. It is not difficult to see that, in the former case, $G_\mu$ admits an L-shape drawing and that, in the latter case, $G_\mu$ admits both a pipe drawing and a rectangular drawing. In the remainder of the section, we consider the case in which $\mu$ has both Q-node and P-node children.

We first show how to construct special square-contact representations of $G_\mu$, that we call {\em canonical drawings}, for any P-node $\mu$ in $T$, assuming that valid drawings have been computed for each S-node child of $\mu$. We distinguish five possible canonical drawings, depending on 
\begin{inparaenum}
\item the number and type of the S-node children of $\mu$ and
\item the presence of edge $st$.
\end{inparaenum}
Each canonical drawing has three variants: 
{\tt vertical} (V),
{\tt horizontal} (H), and
{\tt diagonal} (D).
We name such canonical representations {\em $XY$ drawings}, where $X \in \{V,H,D\}$ denotes the variant of the representation and $Y=1$, if $st \in E(G_\mu)$, and $Y=0$, otherwise. Canonical drawings share the following main property (which, in fact, also holds for valid drawings).

\begin{property}\label{prop:adjacency}
Let $\Gamma_\mu$ be a valid drawing or a canonical drawing of $G_\mu$. Then,
for each vertex $v$ in $V(G^-_\mu)$, it holds that $vs \in E(G_\mu)$ ($vt \in E(G_\mu)$) if:
\begin{enumerate}
\item $S(v)$ has a side that is collinear with a side of $S(s)$ (of $S(t)$) in $\Gamma_\mu$ and 
\item $S(v)$ is separated from $S(s)$ (from $S(t)$) in $\Gamma_\mu$ by the line passing through such a side.
\end{enumerate}
\end{property}

\autoref{prop:adjacency} allows us to modify canonical and valid drawings by appropriate 
\NPAR(s,\hatsymbol)-scaling and \NPAR(t,\hatsymbol)-scaling transformations, with $\hatsymbol \in \{\UL{},\UR{},\DR{}, \DL{}\}$, preserving \mbox{adjacencies between vertices~in~$G_\mu$.}

First, consider a P-node $\mu$ in $T$ with terminals $s$ and $t$ such that $st \notin E(G_\mu)$ and let $\mu_1,\dots,\mu_k$ be the S-node children of $\mu$.
We say that a square-contact representation $\Gamma_\mu$ of $G_\mu$ is an {\em H0 drawing} or a {\em V0 drawing}, if it satisfies the following conditions (in addition to \autoref{prop:adjacency}); refer to \autoref{fig:H}.
\begin{figure}[tb]
	\centering
	\captionsetup[subfigure]{labelformat=empty}
	\subcaptionbox{H0}
	{\includegraphics[height=.23\columnwidth,page=5]{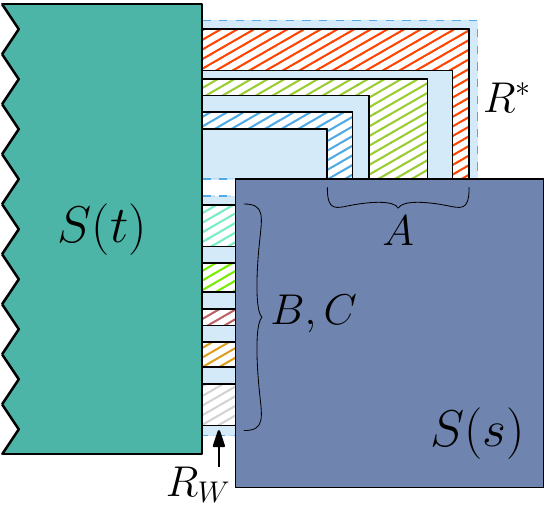}}
	\hfil
	\subcaptionbox{V0}
	{\includegraphics[height=.23\columnwidth,page=12]{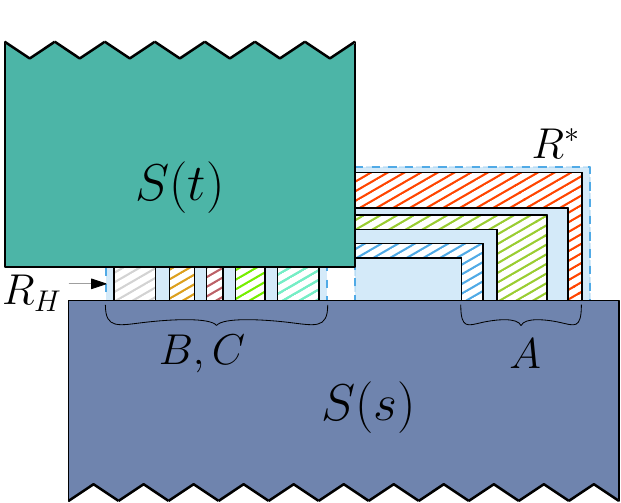}}
		\hfil
	\subcaptionbox{H1}
	{\includegraphics[height=.23\columnwidth,page=7]{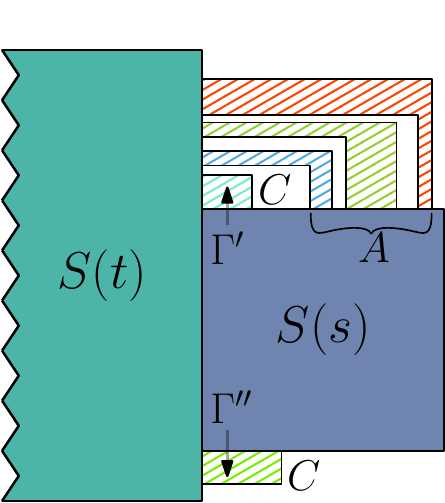}}
  \\
  \vspace{-6mm}
  \subcaptionbox{H1$^\diamond$}
  {\includegraphics[height=.23\columnwidth,page=14]{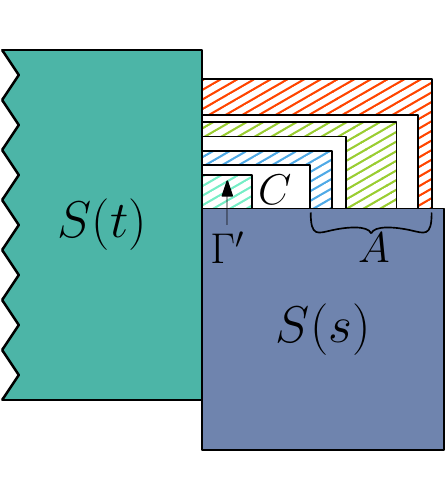}}
  \hfil
	\subcaptionbox{V1}
	{\includegraphics[height=.23\columnwidth,page=9]{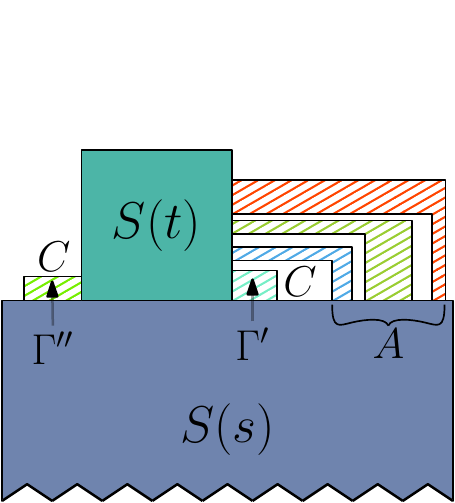}}
  \hfil
  \subcaptionbox{D1}
  {\includegraphics[height=.23\columnwidth,page=10]{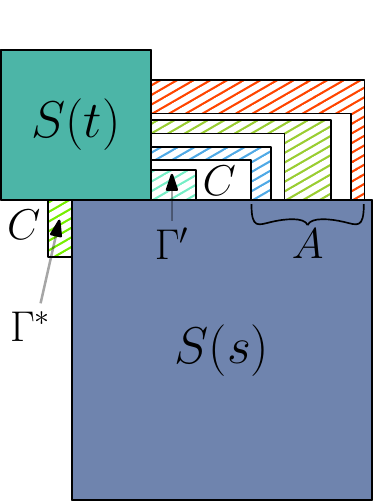}}
	\hfil
	\subcaptionbox{D1$^\diamond$}
	{\includegraphics[height=.23\columnwidth,page=11]{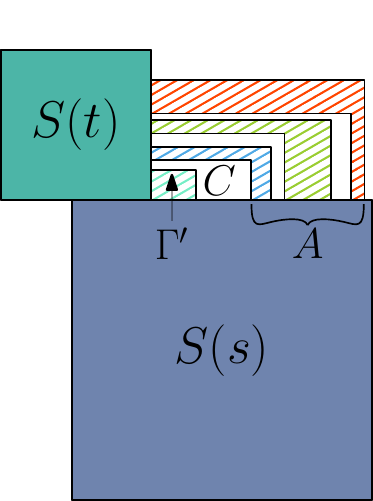}}
	\caption{Canonical drawings of a P-node $\mu$.
The striped regions correspond to L-shapes, horizontal pipes, and rectangles enclosing the square-contact representations of graphs $G^-_{\mu_i}$, for each S-node child $\mu_i$ of $\mu$. Labels $A$, $B$, and $C$ indicate the type of each S-node.
}
	\label{fig:H}
\end{figure}
\medskip
\begin{description}
\item[H0 drawing] 
$S(t)$ lies to the left of $S(s)$, the bottom side of $S(s)$ lies below the bottom side of $S(t)$, and the drawing of $G^-_\mu$ in $\Gamma_\mu$ lies to the right of~$S(t)$, below the top side of $S(t)$, above the bottom side of $S(s)$, and to the left of the right side of $S(s)$. 

\item[V0 drawing] 
$S(t)$ lies above $S(s)$, the left side of $S(s)$ lies to the right of the left side of $S(t)$, and the drawing of $G^-_\mu$ in $\Gamma_\mu$ lies above $S(s)$, to the right of the left side of $S(s)$, below the top side of~$S(t)$, and to the left of the right side of $S(s)$.
\end{description}
\medskip
Now, consider a P-node $\mu$ in $T$ with terminals $s$ and $t$ such that $st \in E(G_\mu)$ and let $\mu_1,\dots,\mu_k$ be the S-node children of $\mu$.
We say that  a square-contact representation $\Gamma_\mu$ of $G_\mu$ is an {\em H1 drawing}, an {\em H1$^\diamond$ drawing}, a {\em V1 drawing}, a {\em D1 drawing}, or a {\em D1$^\diamond$ drawing},  if it satisfies the following conditions (in addition to \autoref{prop:adjacency});  refer to \autoref{fig:H}.
\medskip
\begin{description}
\item[H1 drawing] 
$S(t)$ lies to the left of $S(s)$, the bottom side of $S(s)$ lies above the bottom side of $S(t)$, and the drawing of $G^-_\mu$ in $\Gamma_\mu$ lies to the right of~$S(t)$, below the top side of $S(t)$, above the bottom side of $S(t)$, and to the left of the right side of $S(s)$.  
\item[H1$^\diamond$ drawing] $S(t)$ lies to the left of $S(s)$, the bottom side of $S(s)$ lies below the bottom side of $S(t)$, and the drawing of $G^-_\mu$ in $\Gamma_\mu$ lies to the right of~$S(t)$, below the top side of $S(t)$, above the top side of $S(s)$, and to the left of the right side of $S(s)$.  
\item[V1 drawing] 
$S(t)$ lies above $S(s)$ and the drawing of $G^-_\mu$ in $\Gamma_\mu$ lies above~$S(s)$, below the top side of $S(t)$, to the right of the left side of $S(s)$, and to the left of the \mbox{right side of~$S(s)$.} 
\item[D1 drawing] 
$S(t)$ lies above $S(s)$ and the left side of $S(t)$ lies to the left of the left side of $S(s)$, and the drawing of $G^-_\mu$ in $\Gamma_\mu$ lies to the right of the left side of~$S(t)$, below the top side of $S(t)$, above the bottom side of $S(s)$, and to the left of the right side of $S(s)$. 
\item[D1$^\diamond$ drawing] 
$\Gamma_\mu$ is a D1 drawing of $G_\mu$ in which the drawing of $G^-_\mu$ lies \mbox{to the right of~$S(t)$.}
\end{description}

We now present two lemmata for the possible canonical drawings of each P-node $\mu$ in $T$. Recall that, by \autoref{obs:no-forbidden}, we can assume that $\mu$ is not a forbidden P-node. Let $\mu_1,\dots,\mu_k$ be the S-node children of $\mu$.
The general strategy in the proofs of both lemmata consists of 
\begin{inparaenum}
\item computing appropriate valid drawings $\Gamma_{\mu_1},\dots,\Gamma_{\mu_k}$ for the pertinent graphs $G_{\mu_1},\dots,G_{\mu_k}$ of $\mu_1,\dots,\mu_k$, respectively,
\item 
modifying the square-contact representation of $G^-_{\mu_i}$ contained in $\Gamma_{\mu_i}$, for $i=1,\dots,k$, by means of affine transformations, so that 
 representations derived from S-nodes of the same type lie in the interior of the same polygon, and finally 
\item composing the resulting drawings into \mbox{a canonical drawing of~$G_\mu$; refer 
\lipicsarxiv{to~\cite{paperarxiv}.}{to~\autoref{apx:proofs}\red{.1}.}}
\end{inparaenum}

We first consider the case in which $\mu$ does not contain an edge between its terminals. In this case, by \red{Lemmata}~\ref{le:A}, \ref{le:B}, and~\ref{le:C}, we can assume that $\Gamma_{\mu_i}$ is an L-shape drawing, if $\mu_i$ is of \texttt{Type A}, and a pipe drawing, if $\mu_i$ is of \texttt{Type B} or of \texttt{Type C}, for $i=1,\dots,k$.

\begin{restatable}{lemmax}{LemmaEDGE}\label{le:no-intra-terminals-edge}
Let $\mu$ be a P-node in $T$ with terminals $s$ and $t$ such that $st \notin E(G_\mu)$. Then, graph $G_\mu$ admits an H0 drawing and a V0 drawing.
\end{restatable}

\newcommand{\proofoflemmaEDGE}{
In order to obtain an H0 drawing $\Gamma_{H0}$ of $G_{\mu}$ (a V0 drawing $\Gamma_{V0}$ of $G_{\mu}$) we compose the drawings $\Gamma_{\mu_i}$ of $G_{\mu_i}$, with $1 \leq i \leq k$, as depicted in \autoref{fig:H}\red{(H0)} (in \autoref{fig:H}\red{(V0)}). Specifically, we proceed as follows.  
First, we scale all the pipe drawings of the children of $\mu$ of \texttt{Type B} and of \texttt{Type C} in such a way that they have the same width $W$ (the same height $H$). Then, we compose these drawings in such a way for them to fit in a rectangle $R_W$ of width $W$ (a rectangle $R_H$ of height $H$) and such that no two pipes overlap; let $\Gamma_P$ be the resulting drawing.
Similarly, we scale all the L-shape drawings of the children of $\mu$ of \texttt{Type A} in such a way that they fit into a rectangle $R^*$, the left side of each L-shape is incident to the left side of $R^*$, the bottom side of each L-shape is incident to the bottom side of $R^*$, and no two L-shapes overlap; let $\Gamma_L$ be the resulting drawing. 
Then, we draw $u$ and $v$ as squares $S(u)$ and $S(v)$ of appropriate size in $\Gamma_{H0}$ (in $\Gamma_{V0}$) in such a way that $S(t)$ is incident to the left side of $R_W$ (top side of $R_H$), $S(s)$ is incident to the right side of $R_W$ (bottom side of $R_H$). Finally, we insert a scaled copy of $\Gamma_L$ in $\Gamma_{H0}$ (in $\Gamma_{V0}$) such that the left side of each L-shape contained in $\Gamma_L$ is adjacent to the right side of $S(t)$ and the bottom side of each L-shape contained in $\Gamma_L$ is adjacent to the top side of $S(s)$. This concludes the construction of $\Gamma_{H0}$ and of $\Gamma_{V0}$. \xspace{}
%
}


Then, we consider the case in which $\mu$ contains an edge between its terminals. Recall that, by \autoref{obs:edge-implications}, node $\mu$ has no child of \texttt{Type B} and at most two children of \texttt{Type C}.
In particular, node $\mu$ has two children of \texttt{Type C}, if it is bad, and one child of \texttt{Type C}, if it is almost bad.
In this case, by \red{Lemmata}~\ref{le:A} and~\ref{le:C}, we can assume that $\Gamma_{\mu_i}$ is an L-shape drawing, if $\mu_i$ is of \texttt{Type A}, and a rectangular drawing, if $\mu_i$ is of \texttt{Type C}, for $i=1,\dots,k$.

\begin{restatable}{lemmax}{LemmaNOEDGE}\label{le:intra-terminals-edge}
Let $\mu$ be a P-node in $T$ with terminals $s$ and $t$ such that $st \in E(G_\mu)$. Then, graph $G_\mu$ admits 
\begin{itemize}
\item an H1 drawing, a V1 drawing, and a D1 drawing, if $\mu$ is bad, or
\item an H1$^\diamond$ drawing and a D1$^\diamond$ drawing, if $\mu$ is good or almost bad. 
\end{itemize}
\end{restatable}

\newcommand{\proofoflemmaNOEDGE}{
We define rectangle $R^*$ and the drawing $\Gamma_L$ enclosing the drawing of $G^-_{\mu_i}$ in the L-shape drawing $\Gamma_{\mu_i}$ of each child $\mu_i$ of $\mu$ of \texttt{Type A} exactly as in the proof of
\autoref{le:no-intra-terminals-edge}.

We prove the first part of the statement.
Let $\mu_1$ and $\mu_2$ be the two S-node children of $\mu$ of \texttt{Type C} and let $\Gamma_{\mu_1}$ and $\Gamma_{\mu_1}$ be rectangular drawings of $G_{\mu_1}$ and $G_{\mu_2}$, respectively.
We show how to construct an H1 drawing $\Gamma_{H1}$ of $G_\mu$. The construction of a V1 drawing being symmetric. Refer to \red{Figs.}~\ref{fig:H}\red{(H1)} and~\ref{fig:H}\red{(V1)}.
First, we draw $u$ and $v$ as squares $S(u)$ and $S(v)$ of appropriate size such that 
$S(t)$ lies to the left of $S(s)$, the bottom side of $S(s)$ lies above the bottom side of $S(t)$, the top side of $S(s)$ lies below the top side of $S(t)$, and the right side of $S(t)$ is adjacent to the left side of $S(s)$; let $\Gamma_H$ be the resulting drawing.
Then, we insert a scaled copy of $\Gamma_L$ in $\Gamma_H$ such that the left side of each L-shape contained in $\Gamma_L$ is adjacent to the right side of $S(t)$ and the bottom side of each L-shape contained in $\Gamma_L$ is adjacent to the top side of $S(s)$.
We obtain $\Gamma_{H1}$ from $\Gamma_H$ as follows. First, we insert a scaled copy $\Gamma'$ of  $G^-_{\mu_1}$ in $\Gamma_{\mu_1}$ in the interior of $\Gamma_L$ in such a way that the left side and the bottom side of the rectangle enclosing $\Gamma'$ is adjacent to the right side of $S(t)$ and to the top side of $S(s)$, respectively, and $\Gamma'$ does not overlap with any square in $\Gamma_L$.
Finally,
consider the drawing $\Gamma''$
of $G^-_{\mu_2}$ in $\Gamma_{\mu_2}$ after being mirrored with respect to the $x$-axis. 
We insert a scaled copy of $\Gamma''$ in $\Gamma_H$ in such a way that the top side and the left side of the rectangle enclosing $\Gamma''$ is adjacent to the right side of $S(s)$ and to the bottom side of $S(t)$, respectively. Observe that the flip of the drawing of $G^-_{\mu_2}$ with respect to the $x$-axis has been performed to preserve adjacencies in the final drawing, as we are placing $\Gamma''$ below $S(s)$.
We now show how to construct a D1 drawing $\Gamma_{D1}$ of $G_\mu$. Refer to \autoref{fig:H}\red{(D1)}.
First, we draw $u$ and $v$ as squares $S(u)$ and $S(v)$ of appropriate size such that 
$S(t)$ lies above $S(s)$, the left side of $S(t)$ lies to the left of the left side of $S(s)$, the right side of $S(t)$ lies between the left and the right side of $S(s)$, and the bottom side of $S(t)$ is adjacent to the top side of $S(s)$; let $\Gamma_D$ be the resulting drawing.
Then, we extend $\Gamma_D$ by inserting in it a scaled copy of $\Gamma_L$ and a scaled copy $\Gamma'$ of  $G^-_{\mu_1}$ in $\Gamma_{\mu_1}$ as discussed above for constructing an H1 drawing. Observe that $\Gamma_D$ is now a D1$^\diamond$ drawing.
Finally,
consider the drawing $\Gamma^*$
of $G^-_{\mu_2}$ in $\Gamma_{\mu_2}$ after being mirrored with respect to the $y$-axis and rotated by $-90^\circ$. We insert a scaled copy of $\Gamma^*$ in $\Gamma_D$ in such a way that the top side and the right side of the rectangle enclosing $\Gamma^*$ is adjacent to the bottom side of $S(t)$ and to the left side of $S(s)$, respectively. Observe that 
the flip of the drawing of $G^-_{\mu_2}$ with respect to the $y$-axis and the counter-clockwise rotation by $90^\circ$ have been performed to preserve adjacencies in the final drawing, as we are placing $\Gamma^*$ below $S(t)$ and to the left of $S(s)$.

Now, we prove the second part of the statement. Observe that, since $\mu$ is good or almost bad, it has at most one S-node child of \texttt{Type C}, as S-nodes of \texttt{Type C} are critical.
For the construction of an H1$^\diamond$ drawing and of a D1$^\diamond$ drawing we proceed as discussed above for the construction of $\Gamma_{H1}$ and of $\Gamma_{D1}$, respectively. However, in this case, we can simply omit the last step in these constructions, in which we extend $\Gamma_{H1}$ and $\Gamma_{D1}$ with drawings $\Gamma''$ and $\Gamma^*$, respectively. As observed above, the construction of $\Gamma_{D1}$ immediately yields a  D1$^\diamond$ drawing, if $\mu$ is good or almost bad. 
Refer to \autoref{fig:H}\red{(D1$^\diamond$)}.
Instead, in order to obtain an H1$^\diamond$ drawing from $\Gamma_{H1}$, we only need to perform a final negative \NUR(t)-scaling of $\Gamma_{H1}$ so that the bottom side of $S(t)$ lies above the bottom side of $S(s)$. Refer to \autoref{fig:H}\red{(H1$^\diamond$)}. \mbox{This concludes the proof.} 
}


\medskip

We finally turn our attention to the valid drawings of the S-nodes in $T$.
Let $\mu$ be an S-node in $T$ and let $\mu_1,\dots,\mu_k$ be the children of $\mu$ (where the virtual edge $e_i$, corresponding to node $\mu_i$, precedes the virtual edge $e_{i+1}$, corresponding to node $\mu_{i+1}$, from $t$ to $s$ in $\skel_\mu$). 
The next three lemmata immediately imply \autoref{th:construction}. To simplify their proofs, we assume that each child of $\mu$ is a P-node. In fact, the case in which a child of $\mu$ is a Q-node can be treated analogously to that of a P-node containing an edge between its terminals.
The general strategy in the proofs of all three lemmata consists of 
\begin{inparaenum}
  \item computing appropriate canonical drawings
  $\Gamma_{\mu_1},\dots,\Gamma_{\mu_k}$ for the pertinent graphs $G_{\mu_1},\dots,G_{\mu_k}$ of $\mu_1,\dots,\mu_k$, respectively,
\item modifying these drawings, by means of affine transformations, so that the squares corresponding to terminals shared by different children of $\mu$ can be identified without introducing any overlapping between squares corresponding to internal vertices of $G_{\mu_i}$ and $G_{\mu_j}$, with $i \neq j$, and finally
\item composing the resulting drawings into a valid drawing of $G_\mu$.
\end{inparaenum}

\begin{figure}[t!]
   \centering
   \subcaptionbox{\label{fig:typeA-a}}
  {\includegraphics[height=.23\columnwidth,page=1]{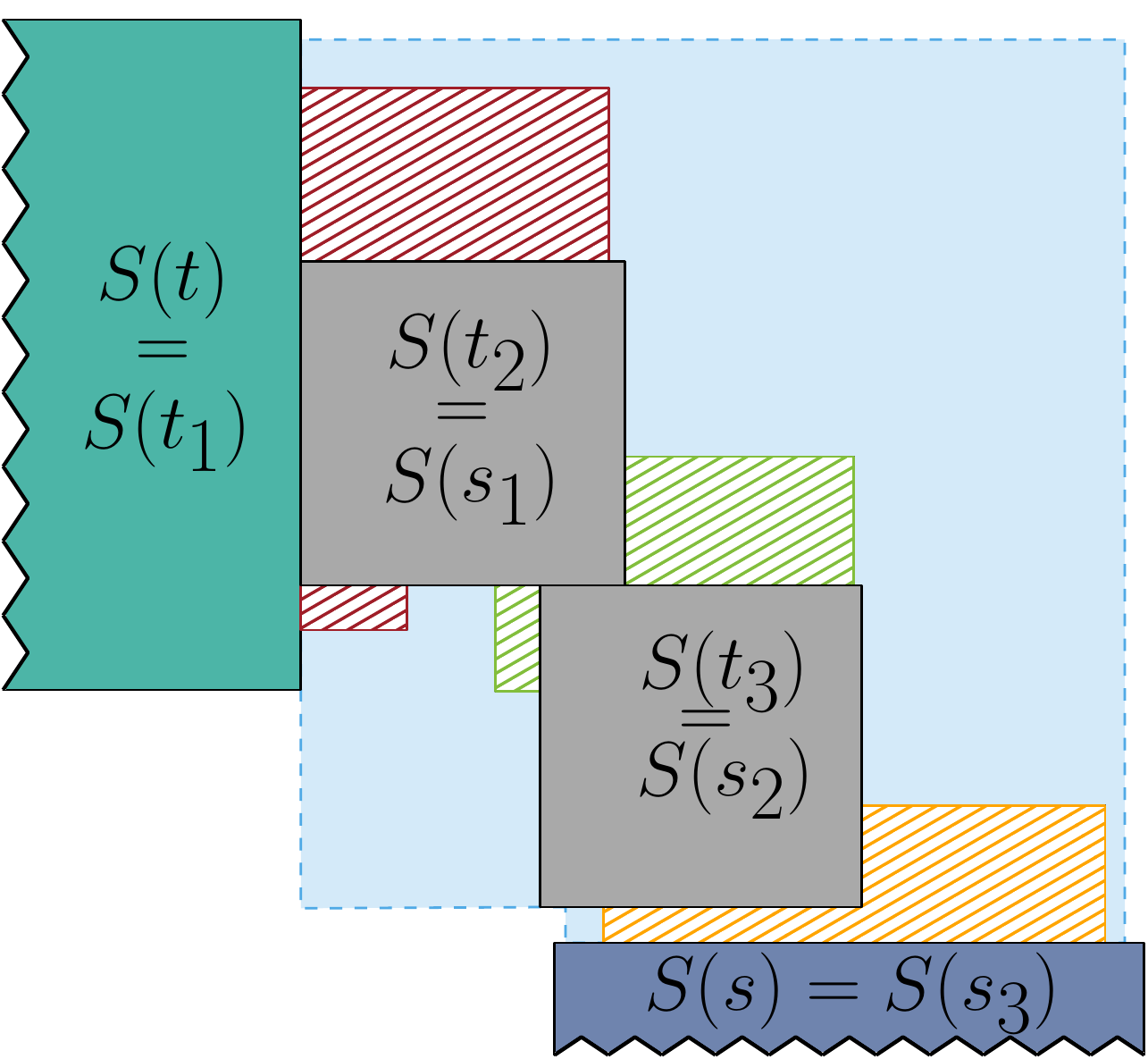}}
  \hfil
  \subcaptionbox{\label{fig:typeA-b}}
  {\includegraphics[height=.23\columnwidth,page=2]{typeABC}}
    \hfil
  \subcaptionbox{\label{fig:typeA-e}}
  {\includegraphics[height=.23\columnwidth,page=5]{typeABC}}
 \\
 \vspace{-14mm}
  \subcaptionbox{\label{fig:typeA-c}}
  {\includegraphics[height=.23\columnwidth,page=8]{typeABC}}
  \hfil
  \subcaptionbox{\label{fig:typeA-d}}
  {\includegraphics[height=.23\columnwidth,page=7]{typeABC}}
  \caption{
Illustrations for the proofs of \red{Lemmata}~\ref{le:A}, \ref{le:B}, and~\ref{le:C}. Striped polygons of the same color enclose different parts of the drawing of each graph $G^-_{\mu_i}$ (contained in the canonical drawing $\Gamma_{\mu_i}$ of $G_{\mu_i}$).
  (a) An H1 drawing of $G_{\mu_1}$, a D1 drawing of $G_{\mu_2}$, and a D1$^\diamond$ drawing of $G_{\mu_3}$ are combined into an L-shape drawing. 
(b) An H0 drawing of $G_{\mu_1}$ and a D1$^\diamond$ drawing of $G_{\mu_2}$ are combined into an L-shape drawing. 
(c) An H1$^\diamond$ drawing of $G_{\mu_1}$ and a D1$^\diamond$ drawing of $G_{\mu_2}$ are combined into a rectangular drawing. 
(d) An H1 drawing of $G_{\mu_1}$ and a D1$^\diamond$ drawing of $G_{\mu_2}$ are combined into a pipe drawing. 
(e) An H1$^\diamond$ drawing of $G_{\mu_1}$ and a D1$^\diamond$ drawing of $G_{\mu_2}$ are combined into a pipe drawing.
    }
\label{fig:typeABC}
\end{figure}

\begin{lemmax}\label{le:A}
If $\mu$ is an S-node of \texttt{Type A}, then $G_\mu$ admits an L-shape drawing.
\end{lemmax}
\begin{proof}
We first describe how to select a valid drawing of $\Gamma_{\mu_i}$ of $G_{\mu_i}$, for $i=1,\dots,k$, based on whether (i) $\length(\mu) > 2$ or (ii) $\length(\mu) = 2$. Recall that, if $\length(\mu) = 2$, then at least one child of $\mu$ does not contain an edge between its terminals, say $\mu_1$ (the case in which $s_1t_1 \in E(G_{\mu_1})$ and $s_2t_2 \notin E(G_{\mu_2})$ is analogous).
\begin{enumerate}[(i)]
\item 
By \autoref{le:no-intra-terminals-edge} and \autoref{le:intra-terminals-edge}, we can construct a drawing $\Gamma_{\mu_i}$, for each $\mu_i$, such that:
\begin{inparaenum}
\item 
$\Gamma_{\mu_1}$ is an H0 drawing, if $s_1t_1 \notin E(G_{\mu_1})$, and $\Gamma_{\mu_1}$ is an H1 drawing (H1$^\diamond$ drawing), if $\mu_1$ is bad (if $\mu_1$ is good or almost bad);
\item 
$\Gamma_{\mu_2}$ is a V0 drawing, if $s_2t_2 \notin E(G_{\mu_2})$, and $\Gamma_{\mu_2}$ is a D1 drawing (D1$^\diamond$ drawing), if $\mu_2$ is bad (if $\mu_2$ is good or almost bad); and
\item 
$\Gamma_{\mu_i}$ is a V0 drawing, if $s_it_i \notin E(G_{\mu_i})$, and $\Gamma_{\mu_i}$ is  a V1 drawing (D1$^\diamond$ drawing), if $\mu_i$ is bad (if $\mu_i$ is good or almost bad), for every $i>2$. 
\end{inparaenum}
\item 
By \autoref{le:no-intra-terminals-edge} and \autoref{le:intra-terminals-edge}, we can construct an H0 drawing $\Gamma_{\mu_1}$ of $G_{\mu_1}$ and a V1 drawing (D1$^\diamond$ drawing) $\Gamma_{\mu_2}$ of~$G_{\mu_2}$, if $\mu_2$ is bad (if $\mu_2$ is good or almost bad). 
\end{enumerate} 

We show how to compose all such drawings into an L-shape drawing $\Gamma_\mu$ of $G_\mu$ as follows. Refer to \autoref{fig:typeABC}\red{(a)} for an example of how to compose drawings $\Gamma_{\mu_i}$, with $i=1,\dots,k$, in case (i) and to \autoref{fig:typeABC}\red{(b)} for an example of how to compose drawings $\Gamma_{\mu_1}$ and $\Gamma_{\mu_2}$ in case (ii).
First, we scale $S(s_i)$ and $S(t_i)$ in $\Gamma_{\mu_i}$ so that the bounding box of the drawing of each connected component of $G_{\mu_i} - \{s_i,t_i\}$ in $\Gamma_{\mu_i}$, for $i=1,\dots,k$, becomes arbitrarily small with respect to the drawing of $S(s_i)$ and $S(t_i)$. This avoids overlapping between internal vertices of $G_{\mu_i}$ and $G_{\mu_j}$, with $i \neq j$, in the next phases of the construction. Then, we scale and translate each drawing $\Gamma_{\mu_i}$ so that $S(t_{i+1})=S(s_{i})$, with $i<k$. 
It is easy to see that, by the choice of the canonical drawings of each $G_{\mu_i}$, there exists a rectangular region in $\Gamma_\mu$ whose interior does not intersect any square representing a vertex in $G^-_\mu$ and whose lower-left corner lies at the intersection point between the vertical line passing through the right side of $S(t)$ and the horizontal line passing through the top side of \mbox{$S(s)$ in $\Gamma_\mu$.} 
\end{proof}

The proof of the next two lemmata also exploits rotations of drawings $\Gamma_{\mu_i}$ and can be carried out in a fashion similar to the proof of \autoref{le:A}. Refer
\lipicsarxiv{to~\cite{paperarxiv}.}{to~\autoref{apx:proofs}\red{.1}.}

\begin{restatable}{lemmax}{LemmaB}\label{le:B}
If $\mu$ is an S-node of \texttt{Type B}, then $G_\mu$ admits a pipe drawing.
\end{restatable}
\vspace{-3.5mm}
\begin{restatable}{lemmax}{LemmaC}\label{le:C}
If $\mu$ is an S-node of \texttt{Type C}, then $G_\mu$ admits a pipe and a rectangular drawing.
\end{restatable}

\newcommand{\proofofLemmaB}{
Recall that $\length(\mu)=2$, at least one of the children of $\mu$ is a bad P-node, say $\mu_1$, and the other child contains an edge between its terminals. The case in which $\mu_2$ is bad and $\mu_1$ is not bad and contains an edge between its terminals can be treated symmetrically.

By \autoref{le:intra-terminals-edge}, we can construct a drawing
$\Gamma_{\mu_1}$ of $G_{\mu_1}$ and a drawing $\Gamma_{\mu_2}$  of $G_{\mu_2}$ such that 
$\Gamma_{\mu_1}$ is an H1 drawing and
$\Gamma_{\mu_2}$ is a V1 drawing, if $\mu_2$ is also bad, or a D1$^\diamond$ drawing, if $\mu_2$ is good or almost bad.

We show how to compose $\Gamma_{\mu_1}$ and $\Gamma_{\mu_2}$ into a pipe drawing of $G_\mu$ as follows. Refer to \autoref{fig:typeABC}\red{(d)} for an example where $\Gamma_{\mu_1}$ is an H1 drawing and $\Gamma_{\mu_2}$ is a D1$^\diamond$ drawing. 
First, we replace $\Gamma_{\mu_2}$ with its copy rotated by $-90^\circ$.
Second, we scale $S(s_i)$ and $S(t_i)$ in $\Gamma_{\mu_i}$ so that the bounding box of the drawing of each connected component of $G_{\mu_i} - \{s_i,t_i\}$ in $\Gamma_{\mu_i}$, with $i \in \{1,2\}$, becomes arbitrarily small with respect to the drawing of $S(s_i)$ and $S(t_i)$. Third, we scale and translate drawing $\Gamma_{\mu_1}$ and $\Gamma_{\mu_2}$ so that $S(t_{2})=S(s_{1})$. Finally, we perform an \NDR(s)-scaling of the obtained drawing of $G_\mu$ so that the drawing of $G^-_\mu$ lies above the bottom side of $S(s)$.
By construction and by the choice of the canonical drawings of $G_{\mu_1}$ and $G_{\mu_2}$, the resulting drawing of $G_\mu$ is a pipe drawing. This concludes the proof.\xspace
}
%
%
\newcommand{\proofofLemmaC}{
Recall that $\length(\mu)=2$, $s_1t_1 \in E(G_{\mu_1})$, $s_2t_2 \in E(G_{\mu_2})$, and none of \mbox{$\mu_1$ or $\mu_2$ is~bad.}

By \autoref{le:intra-terminals-edge}, we can construct a drawing $\Gamma_{\mu_1}$ of $G_{\mu_1}$ and a drawing $\Gamma_{\mu_2}$ of $G_{\mu_2}$ such that
$\Gamma_{\mu_1}$ is an H1$^\diamond$ drawing and
$\Gamma_{\mu_2}$ is a D1$^\diamond$ drawing.

We show how to compose $\Gamma_{\mu_1}$ and $\Gamma_{\mu_2}$ into a pipe drawing $\Gamma_P$ of $G_\mu$; refer to \autoref{fig:typeABC}\red{(e)}.
The first three steps of the construction are exactly as in the proof of \autoref{le:B}. To obtain $\Gamma_P$, we perform an \NDR(s)-scaling and a \NDL(t)-scaling of the obtained drawing of $G_\mu$ so that the bottom side of $S(s)$ and the bottom side of $S(t)$ lie below the bottom side of $S(s_1)=S(t_2)$.

Finally, we show how to compose $\Gamma_{\mu_1}$ and $\Gamma_{\mu_2}$ into a rectangular drawing $\Gamma_R$ of $G_\mu$; refer to \autoref{fig:typeABC}\red{(b)}.  In this case, we obtain $\Gamma_R$ simply by scaling and translating drawing $\Gamma_{\mu_1}$ and $\Gamma_{\mu_2}$ so that $S(t_{2})=S(s_{1})$. 

It is not hard to see that, by construction and by the choice of the canonical drawings of $G_{\mu_1}$ and $G_{\mu_2}$, drawings $\Gamma_P$ and $\Gamma_R$ are a pipe drawing and a rectangular drawing of $G_\mu$, respectively. 
This concludes the proof. \xspace
}

\section{Triconnected Simply-Nested Graphs}\label{se:simply}

In this section, we devote our attention to $3$-connected simply-nested graphs.

A cycle-tree with a single edge removed from the outer cycle is a \emph{path-tree} (to avoid special cases, we allow the outer cycle of the cycle-tree to be a $2$-gon). In path-trees, we refer to vertices in the tree as \emph{tree vertices} and vertices in the external path as \emph{path vertices}. A tree vertex can \emph{see} a path vertex if they share a face in \mbox{the original cycle-tree.}
Define an \emph{almost-triconnected path-tree with root $\rho$, leftmost path vertex $\ell$, and rightmost path vertex $r$} to be a path-tree containing in one of its faces a tree vertex $\rho$ and path vertices $\ell$ and $r$ such that if the edges $\rho\ell$, $\rho r$, and $\ell r$ were added, \mbox{the resulting graph would be a $3$-connected cycle-tree.}

\paragraph{SPQ-decomposition of path-trees.}
We now describe a recursive decomposition for almost-triconnected path-trees. We call this an SPQ-decomposition, because it bears a striking similarity to the SPQ-decomposition of series-parallel graphs. Let $G$ be a $3$-connected cycle-tree, let $\ell r$ be an edge incident to the outer cycle of $G$, and let $\rho$ be a tree vertex incident to the internal face of $G$ edge $\ell r$ is incident to. Also,  let $G'= G - \ell r$ be the almost-triconnected path-tree obtained from $G$ by removing edge $\ell r$. Graph $G'$ defines a rooted decomposition tree $T$ whose nodes are of three different kinds: {\em S-}, {\em P-}, and {\em Q-nodes}.
Each node $\mu$ of $T$ is associated with a path-tree $G_\mu$ with root $\rho_\mu$, leftmost path vertex $\ell_\mu$, and rightmost path vertex $r_\mu$ obtained---except the Q-nodes---from smaller path-trees $T_i$ with root $\rho_i$, leftmost path vertex $\ell_i$, and rightmost \mbox{path vertex $r_i$, for $i=1,\dots,k$, as follows.}

\begin{figure}[b]
\centering
{\includegraphics[scale=0.68,page=2]{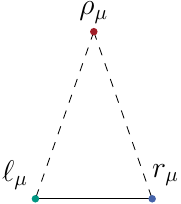}}
\hfil \hspace{5mm}
{\includegraphics[scale=0.68,page=3]{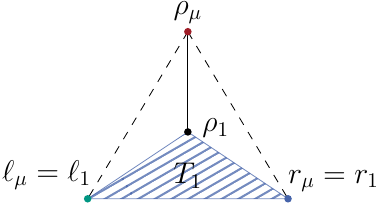}}
\hfil
{\includegraphics[scale=0.68,page=4]{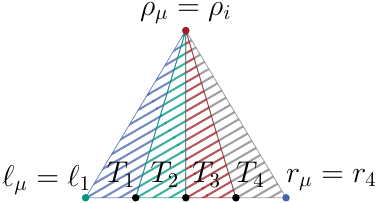}}
\caption{Path-trees associated with a Q-node (left), an S-node (middle), and a P-node (right). Dashed edges may or may not exist.}
\label{fig:nodes}
\end{figure}
\begin{itemize}
\item A {\em Q-node} $\mu$ is associated with a path-tree $G_\mu$ with three vertices: one tree vertex $\rho_\mu$ and two path vertices $\ell_\mu$ and $r_\mu$. The tree vertex $\rho_\mu$ is the root of $G_\mu$, while path vertices $\ell_\mu$ and $r_\mu$ are the leftmost and the rightmost path vertex of $G_\mu$, respectively. Edge $\ell_\mu r_\mu$ will always exist, but edges $\rho_\mu\ell_\mu$ and $\rho_\mu r_\mu$ may or may not exist; see \autoref{fig:nodes}\red{(left)}.
\item An {\em S-node} $\mu$ is associated with a path-tree $G_\mu$ obtained from path-tree $T_1$ by adding a new root $\rho_\mu$ connected to $\rho_1$. Also, $\ell_\mu=\ell_1$ and $r_\mu=r_1$ are the leftmost and the rightmost path vertex of~$G_\mu$, respectively. Edges $\rho_\mu \ell_\mu$ and $\rho_\mu r_\mu$ \mbox{may or may not exist; see \autoref{fig:nodes}\red{(midde)}.}

\item A {\em P-node} $\mu$ is associated with a path-tree $G_\mu$ obtained from path-trees $T_i$ by merging $T_1, T_2, \ldots, T_k$ from left to right as follows. First, roots $\rho_i$ are identified into a new root $\rho_\mu$. Then, the rightmost path vertex $r_i$ of $T_i$ and the leftmost path vertex $\ell_{i+1}$ of $T_{i+1}$ are identified, for $i = 1, \ldots, k - 1$. Path vertices $\ell_\mu=\ell_1$ and $r_\mu = r_k$ are the leftmost and the rightmost path vertex of $G_\mu$, respectively; see \autoref{fig:nodes}\red{(right)}.
\end{itemize}

We have the following lemma. 
\begin{restatable}{lemmax}{theorempathTreeDecomp}\label{thm:path-tree-decomp}
Any almost-triconnected path-tree admits an SPQ-decomposition.
\end{restatable}

\newcommand{\proofTheorempathTreeDecomp}{
We induct on the height of the path-tree~$G$.

If the tree in $G$ has height $1$, then the root is the only tree node. If there are only two path vertices, then $G$ is simply a Q-node. Otherwise, it is a P-node. Let the path in $G$ be $p_1, p_2, \ldots, p_k$. Then $G$ can be produced by merging $k - 1$ Q-nodes. The child $T_i$ of the P-node will be the Q-node consisting of root $r$, left path vertex $p_i$, and right path vertex $p_{i + 1}$.

If the tree in $G$ has height greater than $1$, then first suppose that the only path vertices $root$ can see are $\ell$ and $r$. Then $root$ must have exactly one child in $T$, since otherwise there would be at least one visible path vertex between each pair of adjacent children. So $G$ is an S-node whose single child contains $G-root$. We then continue the decomposition recursively.

Otherwise $root$ can see a path vertex $p \neq \ell,r$. Either there is an edge from $root$ to $p$, or there is space to draw such an edge. The vertices $root$ and $p$ divide $G$ into two components. Let $p_{i_1},\dots,p_{i_k}$ be the left to right sequence of path vertices visible from $root$. For each $j$, the vertices $root$, $p_{i_j}$, and $p_{i_{j+1}}$ define a subgraph $G_j$ containing the path vertices inclusively in between $p_{i_j}$ and $p_{i_{j+1}}$ and the tree vertices on paths in the tree from $root$ to those path vertices.  So $G$ is a P-node with children consisting of $G_j$ for $j =1,\dots,k-1$. Each of these children has a smaller height, so we continue the decomposition recursively. \xspace
}

\lipicsarxiv{In~\cite{paperarxiv}}{In~\autoref{apx:proofs}\red{.2}} we show how to construct a square-contact representation of any almost-triconnected path-tree $G$ without separating triangles and whose outer face is not a triangle by inductively maintaining the invariant depicted in \autoref{fig:invariant} for the S- and P-nodes of an SQP-decomposition of $G$. We formalize this result in the next lemma.

\begin{restatable}{lemmax}{LemmaSCRpathtree}\label{lem:path-tree-scr}
Any almost-triconnected path-tree $G$ without separating triangles and whose outer face is not a triangle
admits a square-contact representation.
\end{restatable}

\newcommand{\proofLemmaSCRpathtree}{
As we construct the path-tree using the SPQ-decomposition, we maintain a square-contact representation for each node. For a path-tree $H$ with root vertex $root$, leftmost path vertex $\ell$, and rightmost path vertex $r$, the square-contact representation for $H$ that we construct obeys the following invariant.

\textbf{Inductive Invariant.}
There are two main cases we consider: either $H$ has two path vertices or $H$ has more than two path vertices. First if $H$ has only two path vertices, then $root$ may be connected to $\ell$ or $r$, but not both. Here we describe the invariant for when $root$ is connected to $r$. The version where $root$ is adjacent to $\ell$ is symmetric swapping $\ell$ with $r$ and left with right throughout. If neither edge is present in $H$, then we can use either version of the invariant.
\begin{enumerate}
\item $S(\ell)$ appears to the left of $S(root)$ and $S(r)$ appears directly to the right of $S(\ell)$ and is below $S(root)$. The bottom of $S(root)$ contacts the top of $S(r)$.
\item The top of $S(\ell)$ is vertically between the bottom of $S(root)$ and one-third the distance to the top of $S(root)$. The bottom of $S(root)$ is vertically above the bottom border of $S(\ell)$.
\item If $root$ and $r$ are adjacent, then all other squares are inside the bounding box defined by extending the lines of the top of $S(\ell)$, the top of $S(r)$, the left of $S(root)$ and the right of $S(\ell)$. We call this the $S(root),S(\ell)$-bounding box.
\end{enumerate}

\begin{figure}
\centering
\includegraphics[scale=0.8]{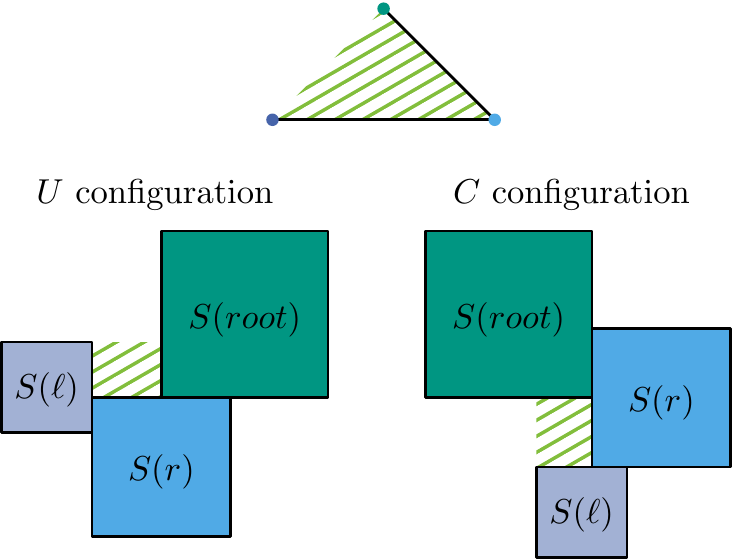}
\caption{The invariant for a cycle-tree with two path vertices.}
\label{fig:invariant-1}
\end{figure}

\begin{figure}
\centering
\includegraphics[scale=0.8]{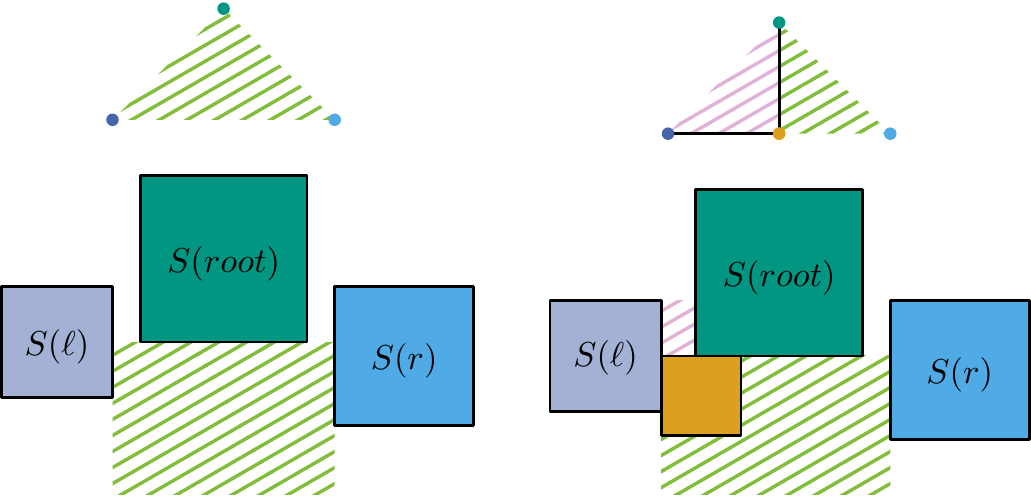}
\caption{The invariant for a cycle-tree with more than two path vertices.}
\label{fig:invariant-2}
\end{figure}

\red{Figure}~\ref{fig:invariant-1} illustrates the two-path-vertex, $root$ adjacent to $r$, case of the invariant. Note that the square-contact representation can be rotated $90^{\circ}$ such that $S(\ell)$ is below $S(root)$ and $S(r)$ appears to the right of $S(root)$ at the top of the square-contact representation. We call the original representation the $U$ orientation and the rotated variant the $C$ orientation for the shape formed by $S(\ell)$, $S(root)$, and $S(r)$. The interior squares are fit into the ``concavity'' of the letter for the orientation.

If $H$ has more than two path vertices, then:
\begin{enumerate}
\item $S(\ell)$ and $S(r)$ have their top on the same horizontal line.
\item $S(\ell)$, $S(root)$, and $S(r)$ appear in that order from left to right with the tops of $S(\ell)$ and $S(r)$ between the bottom of $S(root)$ and one-third the distance to the top of $S(root)$ and the bottoms of $S(\ell)$ and $S(r)$ are below the bottom of $S(root)$.
\item If $\ell$ and $root$ are adjacent, then the right border of $S(\ell)$ touches the left border of $S(root)$. If $\ell$ and $root$ are not connected, then there is a horizontal gap between the right border of $S(\ell)$ and the left border of $S(root)$. This is symmetric for the $r$ and $root$ connection.
\item If $H$ is a P-node and the leftmost (or rightmost) child being merged has only two path vertices, then other squares are allowed inside the $S(root),S(\ell)$-bounding box.
\item All other squares are drawn in the region below $S(root)$, to the right $S(\ell)$, and to the left of $S(r)$.
\item Only interior squares contacting $S(root)$ may touch the horizontal line extending out from the bottom of $S(root)$. Only squares contacting $S(\ell)$ may touch the line extending out from the right of $S(\ell)$. Only squares contacting $S(r)$ may touch the line extending out from the left of $S(r)$.
\end{enumerate}

\red{Figure}~\ref{fig:invariant-2} illustrates the three-path-vertex case of the invariant.
Together the two cases of this invariant guarantees the drawings for S-nodes and P-nodes resembles \autoref{fig:invariant}.


We now show how to construct a square-contact representation $\Gamma$ for a path-tree $G$ by inductively constructing a square-contact representation $\Gamma_x$ that obeys the invariant for each node $x$ in the decomposition. Let $root_x$ be the root vertex of the tree vertices, $\ell_x$ be the leftmost path vertex, and $r_x$ be the rightmost path vertex for a decomposition node $x$. We may drop subscripts when the decomposition node is clear.

\begin{figure}
\centering\includegraphics[scale=0.8]{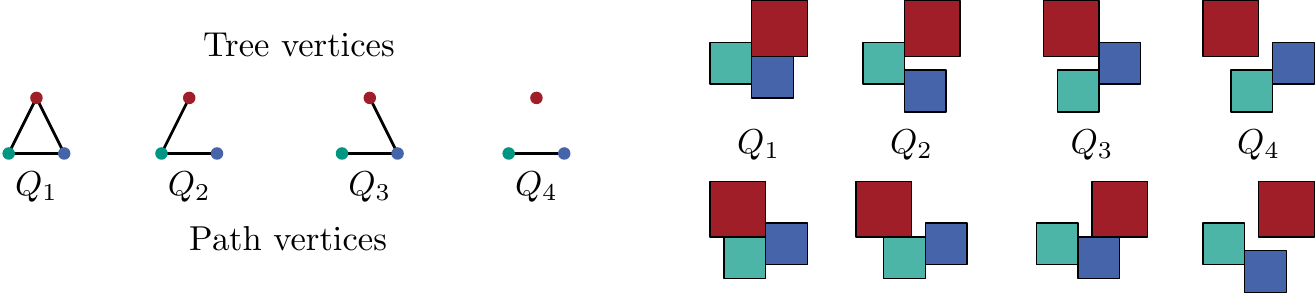}
\caption{The four possible path-trees for a Q-node and 
\mbox{their square-contact representations.}}
\label{fig:qnode}
\end{figure}

First, \red{Figure}~\ref{fig:qnode} shows square-contact representations for each Q-node type obeying the invariant for graphs with two path vertices.

\paragraph{Handling S-nodes}
If $x$ is an S-node that has just $2$ path vertices with child node $c$, then $root_x$ cannot have edges to both $\ell_x$ and $r_x$, otherwise the three vertices form a separating triangle. If $root_x$ is adjacent to $\ell$ choose the orientation of $\Gamma_c$ that places $S_{\Gamma_c}(r)$ below $S_{\Gamma_c}(root_c)$. Then construct $\Gamma_x$ by drawing a square for $x$ of the same size as $S_{\Gamma_c}(root_c)$ directly on top of $S(root_c)$ and perform \NUL(\ell)-scaling until its top is vertically between the bottom of $S(root_x)$ and one-third of the way to the top of $S(root_x)$. $\Gamma_c$ may not have had an adjacency between $root_c$ and $\ell$, in which case we \NUL(root_x)-scale until its left contacts the right of $S(\ell)$. The case when $root_x$ is adjacent to $r$ can be handled symmetrically. Finally when $root_x$ is adjacent to neither $\ell$ nor $r$, we can construct a case where there is a $root_x,\ell$ edge and then perform a slight negative \NUL(root_x)-scaling to remove the extra contact.

If $x$ is an S-node with at least $3$ path vertices and child node $c$, then constructing the square-contact representation is much simpler. $\Gamma_c$obeys the second case of the invariant. Therefore in $\Gamma_c$, $S_{\Gamma_c}(\ell)$ appears on the left of the drawing and $S_(\Gamma_c)(r)$ appears on the right of the drawing with the other squares appearing between them. To construct $\Gamma_x$, start with $\Gamma_c$ and place a new square for $root_x$ on the top of $S(root_c)$ of equal size. If $root_x$ is adjacent to $\ell$, \NUL(root_x)-scale until the left border is vertically in line with the right border of $S(\ell)$. If $root_x$ is adjacent to $r$, \NUR(root_x)-scale until the right border is vertically in line with the left border of $S(r)$. Finally \NUL(\ell)-scale and \NUR(r)-scale until their top borders lie between the bottom border of $S(root_x)$ and one-third of the way to the top of $S(root_x)$. It is possible $root_c$ was adjacent to $\ell$ or $r$ and $root_x$ is not, if so we perform a slight negative \NUL(root_x) or \NUR(root_x)-scaling to remove the contact. 

Every contact in $\Gamma_c$ is preserved in $\Gamma_x$, because every square in $\Gamma_c$ lies below the one-third line of $S_{\Gamma_c}(root_c)$. The only new contacts required in $\Gamma_x$ are those involving $root_x$ which must contact $root_c$ and may contact $\ell$ or $r$. The placement of $S_{\Gamma_x}(root_x)$ guarantees contact with $S_{\Gamma_x}(root_c)$ and the scaling of $S_{\Gamma_x}(root_x)$ introduces contacts with $S_{\Gamma_x}(\ell)$ and $S_{\Gamma_x}(r)$ if needed. By construction $\Gamma_x$ obeys the invariant.

\begin{figure}[t]
\centering
\subcaptionbox{\label{fig:snode}}{
\includegraphics[scale=0.5]{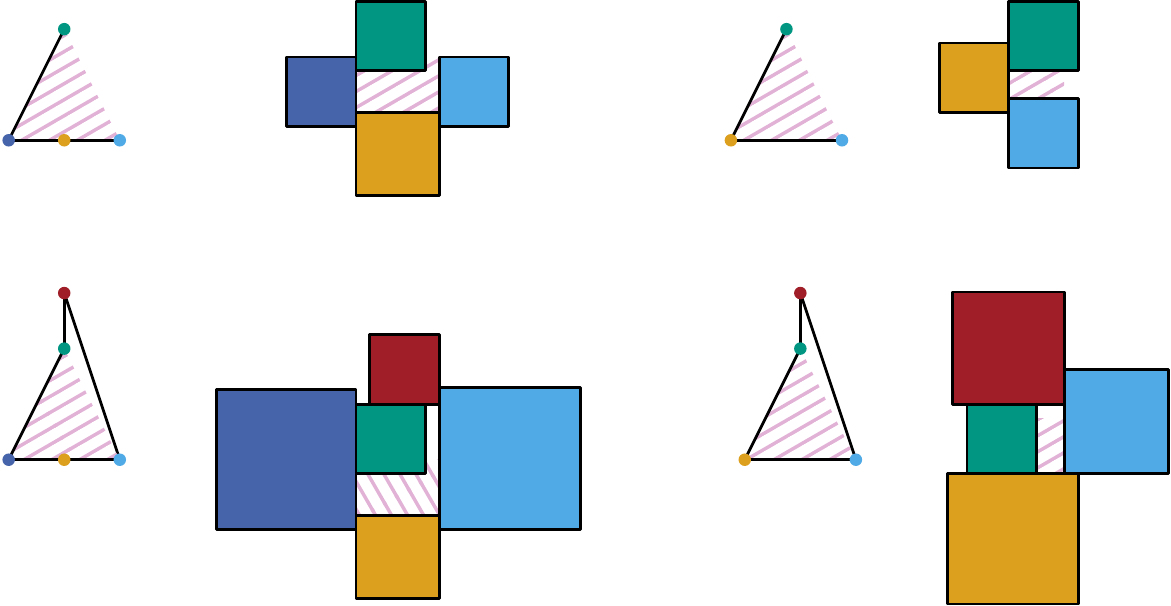}
}
\hfil
\subcaptionbox{\label{fig:pnode}}{
\includegraphics[scale=0.5]{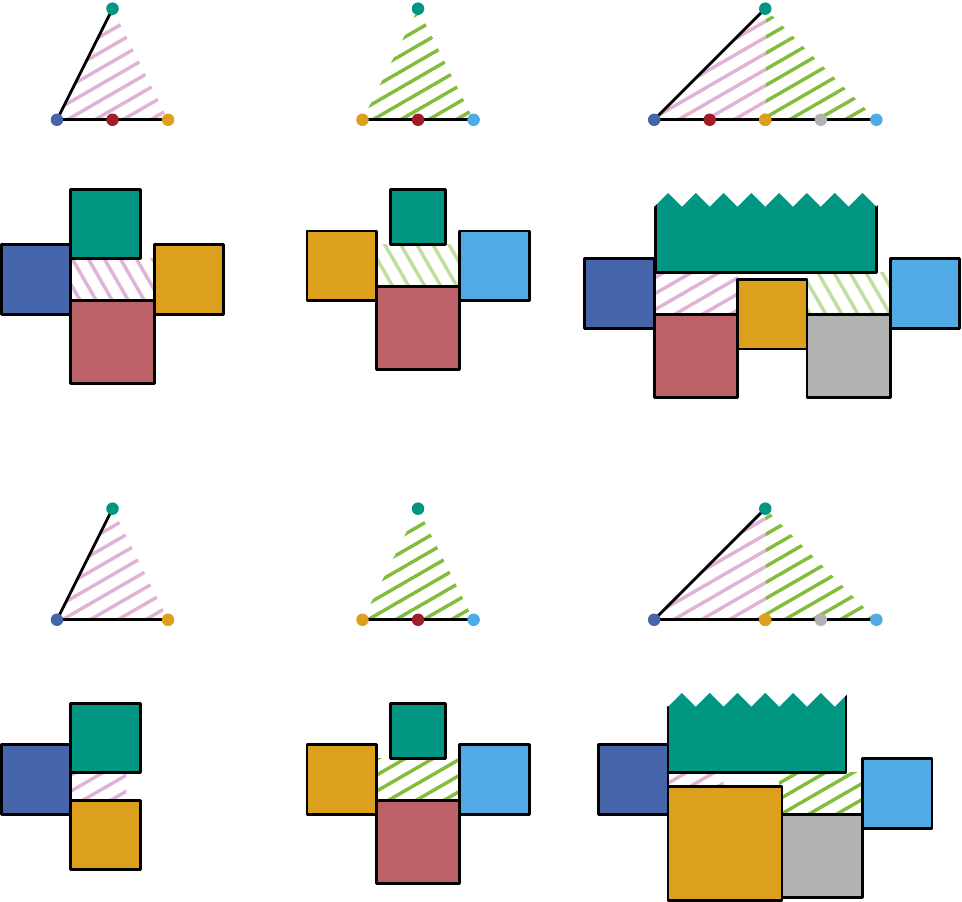}
}
\caption{Examples of the square-contact construction of an S-node (a) and of a P-node (b).}
\label{fig:scr-construction}
\end{figure}

\paragraph{Handling P-nodes}
We handle a P-node $x$ by repeatedly merging the square-contact representations for two adjacent children $c_1$ and $c_2$. For the pair of graphs being merged, call the shared root vertex $root$, the leftmost path vertex $\ell$, the path vertex they share in common $m$, and the rightmost path vertex $r$. We use the $C$ configuration for every child node with just two path vertices, unless it is the leftmost child and has an edge to the right of the two path vertices or it is the rightmost child and has an edge to the left of the two path vertices.

The order we merge children requires some care. For any child in a $C$ configuration, prioritize merging it in the direction of its ``concavity'', that is if $root$ is adjacent to $\ell$ then merge the child with its right sibling. We perform all of these prioritized merges before any others.

After recursively constructing $\Gamma_{c_1}$ and $\Gamma_{c_2}$, scale the two representations such that $S_{\Gamma_{c_1}}(m)$ is the same size as $S_{\Gamma_{c_2}}(m)$. 

If neither or both of $\Gamma_{c_1}$ and $\Gamma_{c_2}$ are $C$ configurations, we perform additional scaling to guarantee the bottom of $S_{\Gamma_{c_1}}(root)$ and $S_{\Gamma_{c_2}}(root)$ are at the same height. To do so, take the side with the higher square bottom for $root$ and perform \NDR(m)-scaling and then rescale down the drawing to keep $S_{\Gamma_{c_1}}(m)$ the same size as $S_{\Gamma_{c_2}}(m)$. Now replace the squares for $root$ with a new square sharing bottom left corners with $S_{\Gamma_{c_1}}(root)$ and bottom right corners with $S_{\Gamma_{c_2}}(root)$. In some cases, this new square overlaps the top interior of $S(m)$. So if it is needed, we translate $S(m)$ downwards either until its top is in line with the bottom of $S(root)$, if $m$ and $root$ are adjacent, or until it is slightly below $S(m)$, if $m$ and $root$ are not adjacent. The resulting square-contact representation has at least three path vertices and obeys the second case of the invariant.

Because of our merge ordering, if there is only one $C$ configuration then $root$ must not be adjacent to $m$. If we attempted to perform the above procedure with such a $C$ configuration, then the result would not obey our invariant because the square of an outer path vertices would be below $S(root)$.

When only one of $\Gamma_{c_1}$ and $\Gamma_{c_2}$ is a $C$ configuration and $root$ is not adjacent to $m$, after overlaying $S_{\Gamma_{c_1}}(m)$ and $S_{\Gamma_{c_2}}(m)$ it is not easy to replace the $root$ squares with one larger square. Without loss of generality we assume $\Gamma_{c_1}$ is the $C$ configuration. In this case, we observe that because $root$ and $m$ are not adjacent, $S_{\Gamma_{c_2}}(m)$ can be translated downwards to slightly below $S_{\Gamma_{c_2}}(root)$ without disturbing any square contacts. After performing this translation, we can \NDR(m)-scale followed by rescaling down $\Gamma_{c_1}$ to keep $S_{\Gamma_{c_1}}(m)$ and $S_{\Gamma_{c_2}}(m)$ the same size such that the bounding box containing the squares of interior vertices has height equal to the slight vertical distance between $S(m)$ and $S(root)$ in the other drawing. Now we can finish the construction by creating a new large square for $root$ in the same manner as the previous case.

The line $S(m)$ is scaled along in the $P$ node construction is guaranteed by the invariant to only intersect squares that are already in contact with $S(m)$. Similarly the bottom of the new $S(root)$ sits along the line touched only by the tops of squares of interior vertices adjacent to $root$. The translation of $S(m)$ is also designed such that any square previously in contact with it, will still be in contact after the translation. Therefore $\Gamma_x$ is a square-contact representation for $x$. By construction $\Gamma_x$ follows the invariant.

\red{Figure}~\ref{fig:scr-construction} shows our construction for two S-nodes on the left and two P-nodes each with two children on the right.\xspace
}

\begin{wrapfigure}[12]{R}{.35\textwidth}
   \centering
    \ifarxive
\vspace{-8mm}
\else
\vspace{-6mm}
\fi
   \includegraphics[scale=.55]{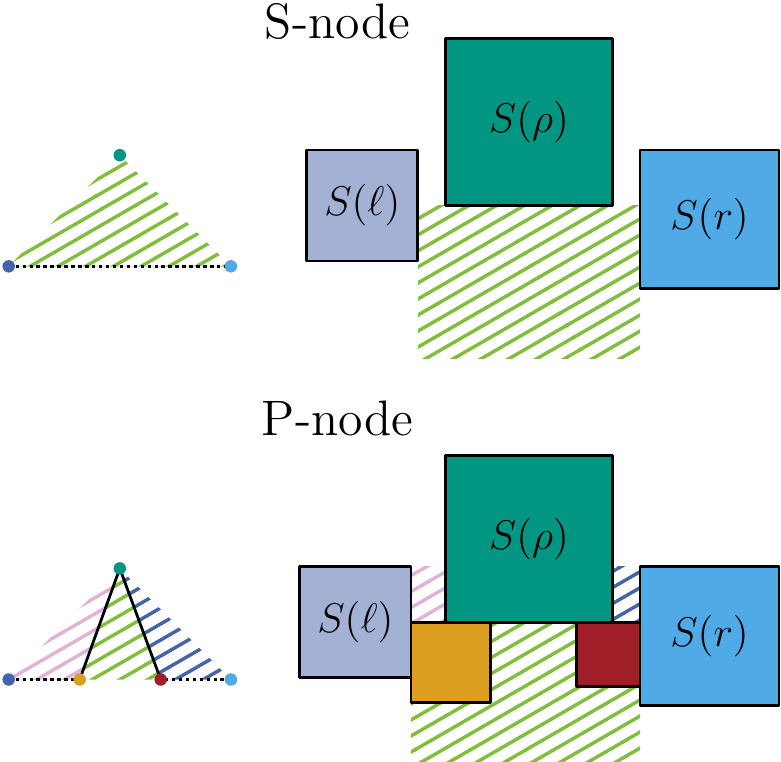}
   \caption{Invariants for S- and P-nodes with more than two path vertices.}
   \label{fig:invariant}
\end{wrapfigure}
To construct a square-contact representation for a $3$-connected cycle-tree, it is natural to remove an edge in the outer cycle to obtain a path-tree, use \autoref{lem:path-tree-scr} to construct a square-contact representation, and then attempt to reintroduce a contact for the removed edge. However, because \autoref{lem:path-tree-scr} places the leftmost and rightmost path vertices on the left and right side of the drawing, it is unclear how to add a contact between them. Instead, we split the cycle-tree into two overlapping almost-triconnected path-trees, obtain their square-contact representations by \autoref{thm:path-tree-decomp}, and overlay them to form a square-contact representation for the entire cycle-tree.

\begin{restatable}{theorem}{TheoremSCRcycletree}\label{le:tree}
Any $3$-connected cycle-tree $G$ without separating triangles and whose outer face is not a triangle admits a square-contact representation.
\end{restatable}

\newcommand{\proofofTheoremSCRcycletree}{
Given a $3$-connected cycle-tree graph $G$ without separating triangles, we split it into two path-trees using the following method:

\begin{enumerate}
\item Root the tree in $G$ arbitrarily.
\item Let $v$ be a leaf vertex of the tree in $G$.
\item If $v$ can see at least three path vertices, select $v$ and a subsequence of three path vertices visible from $v$.
\item Otherwise travel upwards in the tree until reaching a vertex $u$ that can see at least three path vertices
\item There are at most two path vertices connected to descendants of $u$ in the tree. Select $u$ and a subsequence of three path vertices visible from $u$ that include the path vertices connected to the descendants of $u$.
\item Now in either case, we have selected one tree vertex $r$ and three consecutive path vertices which in counter-clockwise order are $p_1$, $p_2$, and $p_3$ visible from $r$. Removing $r$, $p_1$, and $p_3$ disconnects the cycle-tree into two subgraphs $H_1$ and $H_2$ where $H_1$ is the subgraph containing $p_2$.
\item $G-H_1$ and $G-H_2$ are two path-trees that coincide on $r$, $p_1$, and $p_3$.
\end{enumerate}

\red{Figure}~\ref{fig:invariant-1} depicts using this method to obtain two path-trees from a cycle-tree. Because the outer face of $G$ has at least four path vertices, $G-H_2$ has exactly three path vertices, and the two subgraphs share two path vertices, $G-H_1$ has at least three path vertices. The two graphs $G-H_1$ and $G-H_2$ are almost-triconnected path-trees and so by \autoref{thm:path-tree-decomp}, both graphs can be constructed by our decomposition. Then using $r$ as the root of both path-trees, \autoref{lem:path-tree-scr} guarantees $G-H_1$ and $G-H_2$ have square-contact representations $\Gamma_1$ and $\Gamma_2$ respectively satisfying the second case of the invariant.

In particular, $\Gamma_1$ has $S_{\Gamma_1}(r)$ in between $S_{\Gamma_1}(p_3)$ on the left and $S_{\Gamma_1}(p_1)$ on the right while $\Gamma_2$ has $S_{\Gamma_2}(r)$ in between $S_{\Gamma_2}(p_1)$ on the left and $S_{\Gamma_2}(p_3)$ on the right. Our goal is to align the squares for $r$, $p_1$, and $p_3$ in both drawings while not introducing in new square-contacts.

We construct such a square-contact representation $\Gamma=\Gamma(G)$ as follows:
\begin{enumerate}
\item Let $d_i(x,y)$ be the horizontal distance between $S_i(x)$ and $S_i(y)$.
\item Scale the entire drawings of $\Gamma_1$ and $\Gamma_2$ so that $d_1(p_1,p_3) = d_2(p_1,p_3)$.
\item Rotate $\Gamma_2$ by $180^{\circ}$ and vertically align the right sides of the squares $S_{\Gamma_1}(p_1)$ and $S_{\Gamma_2}(p_1)$. Note that the left sides of the squares $S_{\Gamma_1}(p_3)$ and $S_{\Gamma_2}(p_3)$ are also vertically aligned.
\item If $d_1(r,p_1) = d_2(r,p_1)$ and $d_1(r,p_3) = d_2(r,p_3)$, then $S_{\Gamma_1}(r)$ and $S_{\Gamma_2}(r)$ are the same size and we translate $\Gamma_2$ vertically such that $S_{\Gamma_1}(r)$ and $S_{\Gamma_2}(r)$ overlap exactly.
\item If $d_1(r,p_1) \neq d_2(r,p_1)$ and without loss of generality $d_1(r,p_1) < d_2(r,p_1)$, then perform \NDL(r)-scaling in $\Gamma_2$ until the distances are equal. If there were squares in the $S_{\Gamma_2}(r),S_{\Gamma_2}(p_1)$-bounding box, then perform \NUL(a)-scaling on the squares in the bounding box where $a$ is the bottom right corner of the bounding box until the horizontal width of the bounding box is $d_1(r,p_1)$.
\item If $d_1(r,p_3) \neq d_2(r,p_3)$, perform the analogous scaling to make them equal too.
\item Perform \NUL(p_1)-scaling so that its top border is above the top of $S_{\Gamma_2}(p_1)$. Do the analogous scaling for $S_{\Gamma_1}(p_3)$.
\item Remove $S_{\Gamma_2}(r)$, $S_{\Gamma_2}(p_1)$, and $S_{\Gamma_2}(p_3)$.
\end{enumerate}

After the final removal step, $\Gamma$ has one square for each vertex in $G$. We prove for each edge $uv$ in $G-H_1$ that $S_{\Gamma}(u)$ contacts $S_{\Gamma}(v)$ using some case analysis, but first observe that rotations and scalings of entire square-contact representations preserve square contacts so we only need consider cases where $u$ or $v$ was scaled. The only squares undergo scaling in $\Gamma_1$ or $\Gamma_2$ are the squares for $r$, $p_1$, $p_3$, and any vertices in the $S_i(r),S_i(p_j)$-bounding boxes. 

\begin{itemize}
\item Without loss of generality if $u$ is $p_1$ (or $p_3$) then $S_{\Gamma_1}(v)$ touches the right side of $S_{\Gamma_1}(p_1)$. When $S(p_1)$ is scaled, the resulting right border is an extension of the original. The left border of $S(v)$ contains $S_{\Gamma_1}(v)$, because any scaling $S_{\Gamma_1}(v)$ underwent (possibly because $v=r$ or $v$ is in the $S_{\Gamma_1}(r),S_{\Gamma_1}(p_1)$-bounding box) holds at least one point in contact with $S_{\Gamma_1}(p_1)$. 
\item If $u$ is $r$ and $v$ is not $p_1$ or $p_3$, then $v$ may be below $r$ or in the $S_{\Gamma_1}(r),S_{\Gamma_1}(p_1)$-bounding box. If $v$ is below $r$, then the contact is also preserved, because whenever $S(r)$ is scaled the bottom is a superset of the previous bottom.
\end{itemize}

This argument is nearly symmetric for an edge in $G-H_2$. The only difference to observe is that the squares $S_{\Gamma_2}(p_1)$ or $S_{\Gamma_2}(p_3)$ are removed in favor of $S_{\Gamma_1}(p_1)$ and $S_{\Gamma_1}(p_3)$. Other squares only contact $S_{\Gamma_2}(p_1)$ on the right and $S_{\Gamma_2}(p_3)$ on the left. After scaling $S_{\Gamma_1}(p_1)$ ($S_{\Gamma_1}(p_3)$), the right (left) border $S_{\Gamma_1}(p_1)$ ($S_{\Gamma_1}(p_3)$) contains the right (left) border of $S_{\Gamma_2}(p_1)$ ($S_{\Gamma_2}(p_3)$).
Therefore every edge in $G$ has proper contact in $\Gamma$.

We also observe that no new contacts were introduced by these steps. Because for $i=1,2$ and $j=1,3$ the top border of the $S_i(t),S_i(p_j)$-bounding boxes are below the one-third line on $S_i(t)$, the overlaying of $S_{\Gamma_1}(t)$ with $S_{\Gamma_2}(t)$ does not introduce any new contacts (other than overlapping squares for the same vertex).

Thus $\Gamma$ is a proper square-contact representation of $G$.\xspace
}


As Halin graphs are $3$-connected cycle-trees without separating triangles and have, except for $K_4$, a non-triangular outer face, we have the following.

\begin{corollary}\label{co:halin}
Any Halin graph $G \not\simeq K_4$ admits a square-contact representation.
\end{corollary}

Next, we investigate square-contact representations of $2$-outerplanar simply-nested graphs that are not cycle-trees (\autoref{le:twolevels}) and $3$-outerplanar simply nested graphs (\autoref{le:threelevels}).

\begin{theorem}\label{le:twolevels} 
There exists a $3$-connected $2$-outerplanar simply-nested graph that does not admit any proper square-contact representation.
\end{theorem}

\begin{proof}
Consider the two nested quadrilaterals shown in \autoref{fig:nested-quads}\red{(left)}.
One of its two quadrilateral faces must be the outer one, giving the embedding shown.
In any square-contact representation, the inner polygon surrounded by the squares for the four outer vertices must be a rectangle, as it has only four sides. Each of the four inner squares must touch one of the four corners of this rectangle (the corner made by its two outer neighbors). For the four inner squares to touch the four corners of the rectangle and each other, the only possibility is that the rectangle is a square and each inner square fills one quarter of it, as shown in \autoref{fig:nested-quads}\red{(middle)}. However, this representation is improper, as \mbox{diagonally-opposite inner squares meet at their corners.}
\end{proof}

\begin{theorem}\label{le:threelevels} 
There exists a $3$-connected $3$-outerplanar simply-nested graph that does not admit any square-contact representation.
\end{theorem}

\begin{proof}
Consider the graph shown in \autoref{fig:nested-quads}\red{(right)}.
Its quadrilateral face must be the outer one, giving the embedding shown.
As in the proof of \autoref{le:twolevels}, the only possible representation for its two outer quadrilaterals has the four outer squares surrounding a central square region, divided into four quarters representing the four middle vertices, as shown in \autoref{fig:nested-quads}\red{(middle)}. However, this representation leaves no room for the inner vertex.
\end{proof}

\begin{figure}[t]
\centering\includegraphics[scale=0.41]{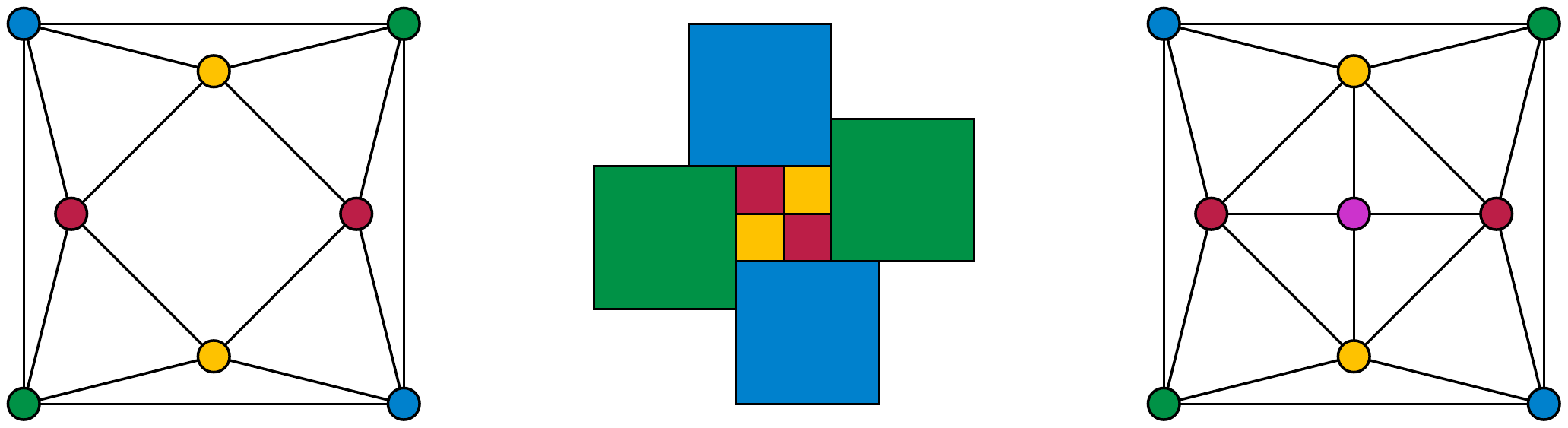}
\caption{Left: Two nested quadrilaterals form a graph with no proper square-contact representation. Middle: An improper square-contact representation for the same graph. Right: A graph with no square-contact representation, even an improper one.}
\label{fig:nested-quads}
\end{figure}

We remark that the graph of \autoref{le:threelevels} is actually $2$-outerplanar simply-nested, but not with its quadrilateral face as the outer face.

\section{Conclusions}\label{se:conclusions}

In this paper, we provided simple characterizations for two notable families of planar graphs that admit proper square-contact representations.~Moreover, we introduced a new decomposition for an interesting family of polyhedral graphs that generalize the Halin graphs, i.e., the $3$-connected cycle-trees. Finally, we showed that the absence of separating triangles and a non-triangular outer face do not guarantee the existence of \mbox{weak and proper square-contact} representations of $3$-outerplanar and $2$-outerplanar simply-nested graphs, respectively.

\smallskip
\paragraph{Acknowledgements.} We thank Jawaherul M. Alam for useful discussions on this subject.

\bibliography{bibliography}
\ifarxive
\clearpage

\appendix
\section{Omitted Definitions}\label{apx:definitions}

A graph is \emph{connected} if it contains a path between any
two vertices.  A \emph{cutvertex} is a vertex whose removal
disconnects the graph.  A {\em separation pair} is a pair of
vertices whose removal disconnects the graph.  A connected graph is
\emph{$2$-connected} if it does not have a cutvertex and a $2$-connected
graph is \emph{$3$-connected} if it does not have a separation pair. The maximal $2$-connected components of a graphs are its {\em blocks}.

A graph is {\em planar} if it admits a drawing in the plane without edge crossings.
A {\em combinatorial embedding} is an equivalence class of planar drawings, where two drawings of a graph are {\em equivalent} if they determine the same circular orderings for the edges around each vertex.
A planar drawing partitions the plane into topologically connected regions, called {\em faces}. 
The bounded faces are the {\em inner faces}, while
the unbounded face is the {\em outer face}.
A combinatorial embedding together with a choice for the outer face defines a {\em planar embedding}.
An \emph{embedded graph} (\emph{plane graph}) is a planar graph with a fixed combinatorial embedding (fixed planar embedding). 

\section{Omitted Proofs}\label{apx:proofs}
In this appendix, we give full details of omitted or sketched proofs.

\subsection{Omitted Proofs of \autoref{se:sp}}\label{apx:proofs-sp}



\Lemmatwotreetosp*

\begin{proof}
\proofoflemmatwotreetosp
 \end{proof}

\LemmaEDGE*

\begin{proof}
\proofoflemmaEDGE
 \end{proof}

 \LemmaNOEDGE*

\begin{proof}
\proofoflemmaNOEDGE
 \end{proof}

 \LemmaB*

\begin{proof}
\proofofLemmaB
 \end{proof}

 \LemmaC*

\begin{proof}
\proofofLemmaC
 \end{proof}

\subsection{Omitted Proofs of \autoref{se:simply}}\label{apx:proofs-simply}

In a path-tree, the \emph{interior} vertices are the vertices other than $root$, $\ell$, and $r$. In an almost-triconnected path-trees, every interior vertex has \mbox{degree at least three.}


\theorempathTreeDecomp*


\begin{proof}
\proofTheorempathTreeDecomp
 \end{proof}


\LemmaSCRpathtree*

\begin{proof}
\proofLemmaSCRpathtree
 \end{proof}


\TheoremSCRcycletree*

\begin{proof}
\proofofTheoremSCRcycletree
 \end{proof}

\else
\fi

\end{document}